\documentclass{IEEEtran}
\usepackage[english]{babel}
\usepackage{microtype}
\usepackage{setspace}
\usepackage{blindtext}
\usepackage{siunitx}
\usepackage{titlecaps}
\Addlcwords{or with if in the of for vs a an and to -off off}
\usepackage{comment}

\usepackage[noadjust,nospace]{cite}

\usepackage{graphicx}
\usepackage{subfloat}
\usepackage{fancyhdr} 
\usepackage{color}
\usepackage{epsfig}
\graphicspath{{Images/}{Images/Comparison/}{Images/Error/}{Images/Graph/}}
\usepackage{pgfplots}
\usepackage[justification=raggedright,singlelinecheck=false,size=footnotesize]{caption}
\usepackage[justification=raggedright,singlelinecheck=false,size=footnotesize]{subcaption}

\usepackage{array}
\usepackage{stfloats}
\usepackage[bottom]{footmisc}

\usepackage{amsthm,amsmath,amssymb,amsfonts}
\usepackage{dsfont}
\usepackage{relsize}
\usepackage{nicefrac}
\usepackage{bbm}
\usepackage{mathtools}
\usepackage[mathscr]{eucal}

\usepackage[ruled]{algorithm}
\usepackage{algpseudocode}
\let\labelindent\relax
\usepackage{enumitem}

\usepackage{float}
\usepackage[absolute]{textpos}

\usepackage{url}
\usepackage[bookmarks=true,colorlinks]{hyperref}
\hypersetup{citecolor=blue}
\usepackage[capitalize]{cleveref}

\newcommand\orcidicon[1]{\href{https://orcid.org/#1}{\includegraphics[scale=0.04]{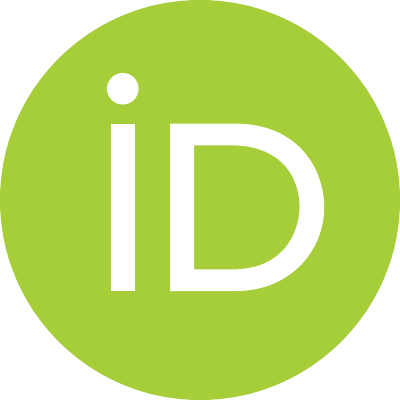}}}

\newcommand{\myParagraph}[1]{{\boldmath \bf \titlecap{#1}.}}
\newcommand{\extended}[2]{#2}

\newcommand{\tradeoff}{competition-collaboration trade-off\xspace}
\newcommand{\node}{agent\xspace}
\newcommand{\nodes}{agents\xspace}

\newcommand{\eg}{\emph{e.g.,}\xspace}
\newcommand{\ie}{\emph{i.e.,}\xspace}
\newcommand{\vs}{\emph{vs.}\xspace}

\newcommand{\Real}[1]{ { {\mathbb R}^{#1} } }

\DeclareMathOperator*{\argmax}{arg\,max}
\DeclareMathOperator*{\argmin}{arg\,min}

\newcommand{\tr}[1]{\mathrm{Tr}\left(#1\right)}
\newcommand{\gauss}{\mathcal{N}}

\newcommand{\E}[2][]{\mathbb{E}_{#1}\left[#2\right]}
\newcommand{\norm}[1]{\left\|#1\right\|^2}
\newcommand{\Enorm}[2][]{\mathbb{E}_{#1}\left[\norm{#2}\right]}
\newcommand{\inv}{^{-1}}
\newcommand{\summat}[1]{\mathcal{B}\left(#1\right)}
\renewcommand{\d}[1]{\mathrm{d}#1}

\newcommand{\one}{\mathds{1}}
\newcommand{\onereg}{\mathds{1}_R}

\newcommand{\ereg}{e_\regSet}
\newcommand{\eregn}{e_{\regSet,n}}
\newcommand{\eregv}{e_{\regSet,v}}

\theoremstyle{plain}
\newtheorem{thm}{Theorem}
\newtheorem{cor}{Corollary}
\newtheorem{prop}{Proposition}
\newtheorem{lemma}{Lemma}
\theoremstyle{definition}

\newtheorem{ass}{Assumption}
\theoremstyle{remark}
\newtheorem{rem}{Remark}

\Crefname{ass}{Assumption}{Assumptions}
\Crefname{prop}{Proposition}{Propositions}

\newcommand{\sensSet}{\mathcal{V}}
\newcommand{\malSet}{\mathcal{M}}
\newcommand{\regSet}{\mathcal{R}}
\newcommand{\xnode}[2]{x_{#1}(#2)}

\newcommand{\xreg}{x_\regSet}
\newcommand{\xmal}{x_\malSet}

\newcommand{\neigh}[1]{\mathcal{N}_{#1}}
\newcommand{\priornode}[1]{\theta_{#1}}
\newcommand{\priorall}{\theta}

\newcommand{\priorallout}{\tilde{\theta}}

\newcommand{\thbar}{\bar{\theta}}

\newcommand{\thbarreg}{\bar{\theta}_\regSet}
\newcommand{\thbarmal}{\bar{\tilde{\theta}}_\malSet}
\newcommand{\Varn}{\Sigma}
\newcommand{\Var}{\widetilde{\Sigma}}
\makeatletter
\newcommand{\varnode}[2][\@empty]{\sigma_{#1#2}%
	\ifx\@empty#1 ^2 \else \fi}
\makeatother
\newcommand{\Creg}{C_R}
\newcommand{\selreg}{S_\regSet}
\newcommand{\selmal}{S_\malSet}
\newcommand{\Wreg}{W_\regSet}
\newcommand{\Wmal}{W_\malSet}

\newcommand{\Gram}[1]{\mathcal{W}_{#1}}
\newcommand{\lami}{\lambda_i}
\newcommand{\lam}{\lambda}
\newcommand{\noise}[2][]{n_{#1}(#2)}

\newcommand{\xregbar}{\bar{x}_\regSet}

\newcommand{\sumall}[1]{\sum_{#1\in\sensSet}}
\newcommand{\sumneigh}[2]{\sum_{#1\in\neigh{#2}}}
\newcommand{\sumreg}[1]{\sum_{#1\in\regSet}}
\newcommand{\summal}[1]{\sum_{#1\in\malSet}}

\addto\captionsenglish{\renewcommand{\figurename}{Fig.}}
\addto\captionsenglish{\renewcommand{\tablename}{Table}}

\newcommand{\LB}[1]{{\color{blue}LB: #1}} 
\newcommand{\marginnote}[1]{\marginpar{\scriptsize\vspace{-\baselineskip}\singlespacing\color{red}#1}}
\newcommand{\blue}[1]{{\color{blue}#1}}
\newcommand{\red}[1]{{\color{red}#1}}
\newcommand{\revision}[1]{{#1}}
\newcommand{\review}[1]{#1}

\newcommand{\linkToPdf}[1]{\href{#1}{\blue{(pdf)}}}
\newcommand{\linkToPpt}[1]{\href{#1}{\blue{(ppt)}}}
\newcommand{\linkToCode}[1]{\href{#1}{\blue{(code)}}}
\newcommand{\linkToWeb}[1]{\href{#1}{\blue{(web)}}}
\newcommand{\linkToVideo}[1]{\href{#1}{\blue{(video)}}}
\newcommand{\linkToMedia}[1]{\href{#1}{\blue{(media)}}}
\newcommand{\award}[1]{\xspace} 
\addto\extrasenglish{}
\addto\extrasenglish{}
\addto\extrasenglish{}

\title{Can Competition Outperform Collaboration? \\ The Role of \review{Misbehaving} Agents}

\author{Luca~Ballotta\textsuperscript{\orcidicon{0000-0002-6521-7142}},~\IEEEmembership{Member,~IEEE}, %
	Giacomo~Como\textsuperscript{\orcidicon{0000-0002-2176-9586}},~\IEEEmembership{Member,~IEEE}, %
	Jeff~S.~Shamma\textsuperscript{\orcidicon{0000-0001-5638-9551}},~\IEEEmembership{Fellow,~IEEE}, %
	and~Luca~Schenato\textsuperscript{\orcidicon{0000-0003-2544-2553}},~\IEEEmembership{Fellow,~IEEE}%
	\thanks{This work has been partially supported
		by the Italian Ministry of Education, University and Research (MIUR) through
		the PRIN project no. 2017NS9FEY entitled ``Realtime Control of 5G Wireless Networks'', 
		the PRIN project 2017 ``Advanced Network Control of Future Smart Grids’’,
		and through the initiative ``Departments of Excellence" (Law 232/2016),
		and by the Compagnia di San Paolo.
		The views and opinions expressed in this work are those of the authors and do not necessarily 
		reflect those of the funding institutions.
	}%
	\thanks{Luca Ballotta and Luca Schenato are with the Department of Information Engineering, University of Padova, 35131 Padova, Italy
		(e-mail: ballotta@dei.unipd.it; schenato@dei.unipd.it).}%
	\thanks{Giacomo Como is with the Department of Mathematical Sciences, Politecnico di Torino, Corso Duca degli Abruzzi 24, 10129, Torino, Italy
		(e-mail: giacomo.como@polito.i).}%
	\thanks{Jeff S. Shamma is the Department Head \& Dobrovolny Chair, Industrial and Enterprise Systems Engineering, University of Illinois Urbana-Champaign,
		USA
		(e-mail: jshamma@illinois.edu).}
}

\begin{document}
	\bstctlcite{bib-options}
	\extended{}{\numberwithin{equation}{section}}
	
	\extended{}{
		\begin{textblock}{20}(-2,0.05)
			\footnotesize
			\centering
			\setstretch{1}
			This article has been accepted for publication in the IEEE Transactions on Automatic Control.\\
			Please cite the article as: L. Ballotta, G. Como, J. S. Shamma, and L. Schenato,\\
			``Can Competition Outperform Collaboration? The Role of \review{Misbehaving} Agents,''\\
			IEEE Transactions on Automatic Control, 2023.\\
		\end{textblock}
	}
	
	\maketitle

\begin{abstract}
	We investigate a novel approach to resilient distributed optimization with quadratic costs
	in a \review{multi-agent system} 
	prone to \review{unexpected events} 
	that make \review{some} \nodes misbehave.
	In contrast to commonly adopted filtering strategies,
	we draw inspiration from
	\review{phenomena modeled through the Friedkin-Johnsen dynamics}
	and argue that adding competition to the mix can improve resilience
	in the presence of \review{misbehaving} \nodes.
	Our intuition is corroborated by analytical and numerical results showing that
	\textit{(i)} \review{there exists} a nontrivial trade-off between full collaboration and full competition and
	\textit{(ii)} our competition-based approach can outperform state-of-the-art algorithms based on \review{Weighted} Mean Subsequence Reduced.
	We also study impact of communication topology and connectivity on \review{resilience},
	pointing out insights to robust network design.
	
\end{abstract}

\section{Introduction}\label{sec:intro}

\IEEEPARstart{W}{ith} great power comes great responsibility,
and networked systems are powerful indeed.
From smart grids managing energy consumption~\cite{5357331,7063267}
to sensors monitoring vast areas~\cite{6951347},
to autonomous cars for intelligent mobility~\cite{9310743,9233967},
everyday life relies evermore on
control of connected devices. 

While this brings numerous benefits, 
a major drawback is that malicious \nodes can locally 
intrude from any point in the system,
and cause serious damage at global scale. 
Recently, Department of Energy secretary
stated that enemies of the United States can shut down the U.S. power grid,
and it is known that hacking groups around the world have high technological sophistication~\cite{USpowergridattack}.
Cyberattacks hit Italian health care infrastructures during the COVID-19,
disrupting services for weeks~\cite{healthcareItahacking}.
Another concern is accidental failures spreading
from single source nodes. 
Cascading failure damages have notable examples, 
from city-wide electricity blackouts
to denial of service of web applications.
Furthermore,
as new frontiers through massively connected devices in Networked Control Systems are breached,
thanks to powerful communication protocols such as 5G,
this problem will only gain in importance.

\subsection{Related Literature}\label{sec:intro-related-literature}
The problems above have been extensively studied in literature. 
A body of work investigates control techniques to overcome
fragility of specific applications.
Examples are power outage in smart grids~\cite{6733531,8613860},
cascading failures in cyber-physical systems~\cite{6915846,8070359,7438924,6911943},
denial of service~\cite{7289347,8629941},
robot gathering~\cite{doi:10.1137/050645221},
and distributed estimation~\cite{4787093},
to name a few.
From a methodological perspective, 
control and optimization literature mostly focuses on robustness
of distributed algorithms and control protocols
to a fraction of misbehaving \nodes.
This approach can tailor either intentionally malicious \nodes,
such as cyber-attackers,
or accidental faults due to,
\eg hardware damage.
A fundamental subclass of such approaches is \emph{resilient consensus},
aimed to enforcing consensus of normally behaving (or \textit{regular})
\nodes in the face of unknown adversaries.
The consensus problem
has been deeply studied in the past decades~\cite{XIAO200733} and underlies
a plethora of application domains.
In particular, \textit{average consensus} is a cornerstone
in distributed estimation~\cite{4787093}
and optimization~\cite{9241497,8015179,7134753},
management of power grids~\cite{8716798}, 
distributed Federated Learning~\cite{8917592,MLSYS2019_bd686fd6},
among others.
Unfortunately,
the standard consensus protocol is fragile 
and misbehaving \nodes can arbitrarily deviate the system trajectory. 
To tame this issue,
the most common approaches rely on the
filtering strategy referred to as \textquotedblleft Mean Subsequence Reduced" (MSR),
whereby \nodes discard suspicious messages (largest and smallest values) from updates~\cite{262588}.
The pioneering paper~\cite{6481629} 
introduced a weighted version (W-MSR) and defined $ r $-\textit{robustness} of graphs,
a suitable index that enables theoretical guarantees for resilient consensus based on W-MSR.
Among the many variants and adaptations of W-MSR,
~\cite{DIBAJI201523} studies resilient control for double integrators,
~\cite{WANG20203409} tackles mobile adversaries,
~\cite{8795564} focuses on leader-follower framework,
~\cite{SHANG2020109288} targets nonlinear systems with state constraints,
\review{~\cite{10018183} extends the notion of $r$-robustness to time-varying graphs,}
and~\cite{7447011,9137648,5605238} consider generic cost functions
to achieve resilience in general distributed optimization.

Other approaches in literature do not filter information from neighbors,
but explore enhanced capabilities of regular \nodes.
For example,
~\cite{8814959} uses a buffer to store all values received from other \nodes and
replaces the thresholding mechanism with a voting strategy
followed by dynamical updates,
~\cite{8186231} studies algorithmic robustness enabled by trusted \nodes,
~\cite{9252111} proposes dynamically switching update rule for continuous-time double integrators,
and~\cite{8798516,10.1117/12.918927,9468419} use stochastic or heuristic trust scores to filter out potentially malicious transmissions,
providing probabilistic bounds on detection, convergence, or deviation from average consensus.
While such approaches may overcome limitations of MSR-based strategies,
they usually require either stronger assumptions on the network (\eg trusted \nodes)
or burdening local computation or storage resources.

\subsection{\titlecap{novel contribution}}

Despite the success of MSR-based strategies, 
a critical point is dependence of theoretical guarantees on $ r $-robustness of the underlying graph,
which allows regular \nodes to reach resilient consensus if such an index is large enough.
In fact,
it is difficult to characterize the steady-state behavior of agents
if some minimal robustness is not met. 
Even though algorithms might practically work,
comprehensive theoretical guarantees are still lacking,
and also,
some applications require more conservative but safer approaches.
In particular, 
while in some cases \nodes may just agree on a common value,
other tasks require \emph{average consensus} to succeed.
Thus,
we depart from classical filtering strategies and
seek a framework for resilience that can offer theoretical guarantees in a broader sense.

\begin{figure}
	\centering
	\includegraphics[width=0.7\linewidth]{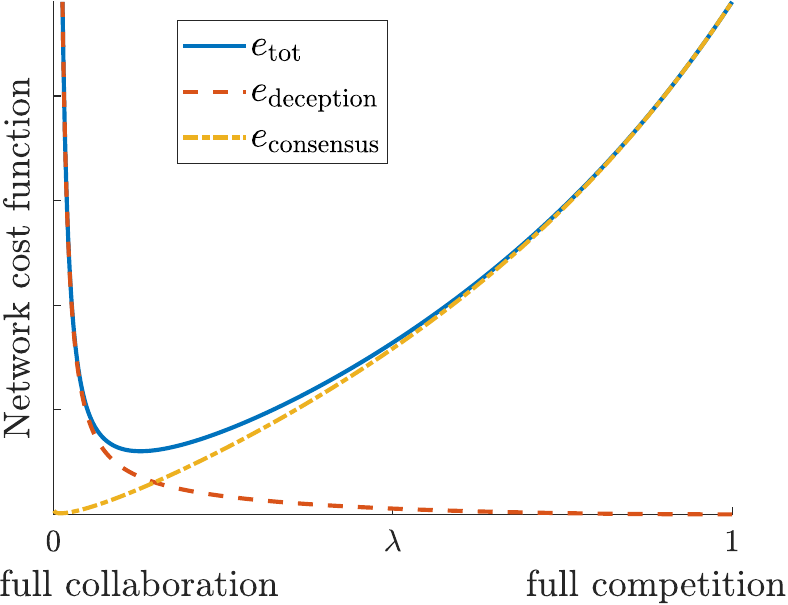}
	\caption{Competition \emph{vs.} collaboration in distributed quadratic optimization.
		The \revision{global cost} $ e_{\text{tot}} $ is the sum of two contributions
		\revision{that reflect two contrasting attitudes of regular \nodes}:
		\revision{$ e_{\text{deception}} $ is caused by (erroneously) trusting misbehaving \nodes,
			which \review{makes them drift away from the nominal average},
			while $ e_{\text{consensus}} $ is due to the \review{competition among regular \nodes},
			\review{which mitigates misbehaviors
			but also prevents regular \nodes from reaching a consensus.}
			The tunable parameter $ \lam \in [0,1] $ allows regular \nodes to smoothly transition from 
			\textit{full collaboration} \review{($\lam=0$)},
			where they fully trust all \nodes in the network,
			to \textit{full competition} \review{($\lam=1$)},
			where they trust \review{no other \node},
			producing a rich range of behaviors at local and global scale.
		}
	}
	\label{fig:trade-off-coll-comp}
\end{figure}

Towards this goal, we set the stage with two key moves.
Firstly,
\review{rather than finding conditions that enforce a consensus among regular \nodes,
	which only indicates if a system \textit{is resilient},
	we aim to measure the \textit{level of resilience},
	which we evaluate through}
the cost of a distributed optimization problem with quadratic costs.
Secondly,
\review{we aim to modify the original problem to make it robust to misbehaving \nodes
rather than adapting a consensus protocol.}
Stepping forward,
we propose an update rule based on the celebrated Friedkin-Johnsen (FJ) dynamics~\cite{FJdynamics}
to enhance resilience \review{of the addressed distributed optimization problem.}
The key feature of the FJ dynamics is a tunable parameter $ \lam \in\ [0,1] $
that allows to smoothly transition from \review{the regime of} \emph{full collaboration} ($\review{\lambda=0}$),
where each regular \node equally trusts all \nodes,
to \review{the regime of} \emph{full competition} ($\review{\lambda=1}$),
where each regular \node regards all others as adversaries.
\review{We refer to the regime with $\lambda\in(0,1]$ as \textit{competition-based}
	because regular \nodes are forced to (partially) mistrust the others.}
This approach allows us to study \review{resilience variations} that arise
from different choices of \nodes that can trust their neighbors or not,
a choice that turns out to be crucial if adversaries are present.
In fact, we observe a fundamental \review{performance trade-off that we name}
\emph{\tradeoff}:
in general,
\textit{\review{the optimal resilient strategy is hybrid,
	namely each regular \node should partially compete with its neighbors}},
as depicted in~\autoref{fig:trade-off-coll-comp}.
The global cost (solid blue)
is the sum of two conflicting contributions
that represent deception due to collaboration with misbehaving \nodes (dashed red)
and inefficiency caused by competition against regular \nodes (dashed-dotted yellow).
\review{To achieve analytical intuition about such a \tradeoff,
we leverage the \textit{social power},
a tool drawn from opinion dynamics
that sheds light on the twofold effect of the parameter $\lambda$ used to instantiate the FJ dynamics.}


After \review{analytically} characterizing the proposed competition-based protocol,
we fix the update rule and shift attention to the network
in order to assess how it impacts \review{resilience} of regular \nodes. 
\review{In particular, we numerically show how network connectivity can mitigate misbehavior
and how the performance varies as the network gets sparser or less balanced.}
In fact, we heuristically observe that not only high connectivity,
but also degree balance across \nodes is useful to tame unknown adversaries,
that intuitively can exploit highly connected areas to quickly spread damage at global level.

Besides new results,
this article extends the preliminary conference version~\cite{ballottaCDC22compColl} in two ways.
Firstly,
we consider a more general prior distribution of the observations of \nodes.
Secondly,
we compare our proposed strategy with both standard W-MSR~\cite{6481629}
and recently proposed SABA~\cite{8814959}.

\subsection{\titlecap{organization of the article}}

We motivate average consensus for distributed optimization in~\autoref{sec:setup},
\review{model} a class of adversaries in~\autoref{sec:mal-agents}, 
\review{and introduce the performance metric used to quantify resilience in~\autoref{sec:performance-metric}.}
In~\autoref{sec:resilient-strategy},
we propose our competition-based protocol:
\review{we introduce the FJ dynamics in~\autoref{sec:FJ-dynamics}},
\review{compute the cost function in~\autoref{sec:computation-consensus-error},}
and formally characterize the cost function and its minimizer
in~\cref{sec:opt-lam-nontrivial,sec:opt-lam-vs-d}.
\review{In~\autoref{sec:numerical-tests},
we report numerical tests that support our analytical intuition}.
Then,
in~\autoref{sec:trade-off},
we offer analytical insight on the \tradeoff using the notion of social power.
In~\autoref{sec:network-optimization},
we \review{numerically explore} the impact of the communication network on \review{resilience}.
To \review{evaluate} our approach, 
we perform simulations in~\autoref{sec:literature-comparison}
and show that it can outperform MSR-based methods.
We conclude by addressing potential avenues for future research in~\autoref{sec:conclusion}.

\section{\titlecap{setup and problem formulation}}\label{sec:setup}

We consider a \review{multi-agent} system composed of $ N $ {\nodes}
\review{labeled as} the set $ \sensSet = \{1,\dots,N\} $.
Each \node $ i\in\sensSet $ carries local information encoded by an {\emph{observation}} $ \priornode{i} \in\Real{}$
and a variable \textit{state} $ x_i\in\Real{} $.
For notation convenience,
we stack all states and observations in the vectors $ x\in\Real{N} $ and $\priorall\in\Real{N}$,
respectively.

Within the network,
some \nodes behave according to a control task at hand,
while others cannot be controlled and may \review{deviate from the task}.
We call the former \nodes \textit{regular} and the latter \nodes \textit{\review{misbehaving}}.
Because the \review{misbehaving \nodes} cannot be involved in cooperative tasks though their uncontrolled nature,
\review{we consider a distributed optimization problem involving only the regular \nodes.
We assume that each {regular \node} 
wishes to adjust its state so as to minimize a quadratic mismatch among all observations,}
%
\begin{equation}\label{eq:agent-cost-regular}
	f_{\text{local}}(x_i) \doteq \sumreg{j} \left(x_i - \priornode{j}\right)^2, \quad i\in\regSet,
\end{equation}
where $ \regSet \subseteq \sensSet $ gathers all {regular \nodes}.
By straightforward calculations,
\eqref{eq:agent-cost-regular} can be rewritten as
\begin{equation}\label{eq:agent-cost-regular-rewritten}
	\begin{aligned}
		f_{\text{local}}(x_i) &= R\left(x_i - \thbarreg\right)^2 - R\review{\thbarreg^2} + \sumreg{j}\priornode{j}^2 \\
	\end{aligned}
\end{equation}
where $ R\doteq|\regSet| $ and
$ \thbarreg $ is the average of observations $ \{\priornode{i}\}_{i\in\regSet} $.


The distributed optimization task is then given by
\begin{equation}\label{eq:network-cost-regular}
	\review{\argmin_x \dfrac{1}{R}\sumreg{i}f_\text{local}(x_i) 
		 					= \argmin_x\sumreg{i}\left(x_i - \thbarreg\right)^2,}
\end{equation}
%
which is solved if and only if all regular \nodes reach \emph{average consensus} among them,
\ie $ x_i = \thbarreg $ for all $ i\in\regSet $.

\review{In the nominal scenario where all \nodes are regular ($\sensSet=\regSet$),
	the cost~\eqref{eq:network-cost-regular} can be minimized
	via the \textit{consensus dynamics} (or \textit{consensus protocol})
	$x(k+1) = W^ox(k)$
	where $x(0)\stackrel{}{=}\theta$ and
	$W^o$ is a doubly stochastic irreducible matrix that leads \nodes to average consensus.
	\review{Interpreting $W^o$ as a (weighted) communication matrix,
	the consensus dynamics allows \node $j$ to communicate its state to \node $i$ if and only if $W_{ij} > 0$.}
	
	
	However,
	the standard consensus protocol easily fails
	in the presence of misbehaving \nodes~\cite{6481629}. 
	We next
	introduce a misbehavior model
	that disrupts the nominal protocol. 
}

\subsection{\review{Misbehaving} Agents}\label{sec:mal-agents}

Misbehaving \nodes follow state trajectories with no relation to optimization task~\eqref{eq:network-cost-regular}
and broadcast potentially misleading information to neighbors.
We denote the subset of misbehaving \nodes by $ \malSet $
with $ M\doteq\lvert\malSet\rvert $,
$ \sensSet = \malSet \cup \regSet $,
and $ \malSet \cap \regSet = \emptyset $. 
\review{Also,
without loss of generality,
we label the \nodes as $ \regSet = \{1,\dots,R\} $ and $ \malSet = \{R+1,\dots,N\} $.
The vectors $\xreg\in\Real{R}$ and $\xmal\in\Real{M}$ stack the states of regular and misbehaving \nodes,
respectively,
with $x^\top = [\xreg^\top \; \xmal^\top]$.}

\review{To address a general scenario and remove dependence on the specific values of observations,
	we assume that these are drawn from a prior distribution.}

\review{\begin{ass}[Distribution of observations]\label{ass:prior-distribution}
		Observations $\{\priornode{i}\}_{i\in\sensSet}$ are distributed as random variables with mean $\E{\priorall}=0$
		and covariance matrix $ \Varn\doteq\E{\priorall\priorall^\top}\succ0 $. 
		\revision{We denote $ \Sigma_{ii} = \varnode{i} $ and $ \Sigma_{ij} = \varnode[i]{j} \, \forall i\neq j $.}
\end{ass}

While the standard consensus and~\cref{ass:prior-distribution} are suited to an ideal scenario,
misbehaving \nodes may disrupt the task~\eqref{eq:network-cost-regular}.
In the following,
we assume that
misbehaving \nodes constantly transmit noisy versions of their observations:
\begin{equation}\label{eq:malicious-dynamics}
	\xnode{m}{k} = \priornode{m} + v_m + \noise[m]{k}, \quad \forall m\in\malSet.
\end{equation}
We refer to the constant input $v_m$ as \textit{(deception) bias} and to the varying input $\noise[m]{k}$ as \textit{(deception) noise}.
In words,
the deception bias $v_m$ makes
the observation $\priornode{m}$ of the misbehaving \node $m$ an outlier w.r.t. the expected range of values of observations as per~\cref{ass:prior-distribution}.
Conversely,
the deception noise $\noise[m]{k}$
hides the true state of the misbehaving \node from its neighbors,
akin purposely injected measurement noise.

\begin{ass}[Misbehavior model]\label{ass:mal-node-dynamics}
	We stack 
	biases in the vector $ v\in\Real{M} $ 
	and noises in the vector $\noise{k}\in\Real{M}$.
	Further,
	we set their statics as 
	$\E{v} = 0$,
	$\E{vv^\top} = V\succeq0$,
	$ \E{v\priorall^\top} = 0$, 
	$\E{\noise{k}} = 0$, 
	$\E{\noise{k}n^\top(h)} = \delta_{kh}Q\succeq0$, 
	and $\E{\noise{k}\priorall^\top} = 0 \ \forall k,h\ge0$,
	where $\delta_{kh}=1$ if $k=h$ and $\delta_{kh}=0$ otherwise.
\end{ass}
}

\begin{rem}[Misbehavior \vs intelligent attacks]\label{rem:misbehavior}
	\cref{ass:mal-node-dynamics} is consistent with a portion of the literature on resilient consensus,
	where algorithms are tested against constant or drifting \review{misbehaving} \nodes
	that steer their neighbors far off the nominal consensus~\cite{8798516,8814959,WANG20203409,9468419}.
	\review{In our case,
	misbehaving \nodes are stubborn on average
	but behave in a less trivial (noisy) way.}
	On the other hand,
	\review{smart (malicious) adversaries}
	may need to be contrasted by \review{sophisticated} strategies~\cite{8567999,8788623}.
	This case is outside the scope of this article,
	where we \review{explore competition as a tool to enhance resilience},
	and we defer \review{a comprehensive study with intelligent attacks} to future work.
\end{rem}
		\review{

\subsection{Performance Metric}\label{sec:performance-metric}

In light of problem~\eqref{eq:network-cost-regular} 
and assuming that the states of regular \nodes are updated by a control protocol overtime,
we use the following performance metric to measure the resilience of the system,
which we refer to as \emph{(average) consensus error}:

\begin{equation}\label{eq:cons-error-regular}
	\ereg \doteq \lim_{k\rightarrow+\infty}\E{\sumreg{i}\left(x_i(k) - \thbarreg\right)^2}.\tag{CE} 
\end{equation}

The error~\eqref{eq:cons-error-regular}
coincides with the objective cost of the optimization problem~\eqref{eq:network-cost-regular}
(up to additive constants that depend only on the observations)
averaged over the stochastic elements within the system dynamics,
such as observations $\priorall$ of all \nodes
and deception biases $v$
and noises $\noise{k}$ of misbehaving \nodes.

While the standard consensus protocol achieves $\ereg=0$ in the nominal scenario where all \nodes reach average consensus,
the presence of unknown misbehaving \nodes makes $\ereg$ grow,
degrading the collaborative task~\eqref{eq:network-cost-regular}.
In the following section,
we propose an update protocol that makes regular \nodes more resilient
by decreasing the consensus error $\ereg$ and hence improving the performance associated with the task~\eqref{eq:network-cost-regular}.
}
	\extended{\newpage}{}

\section{\titlecap{resilient average consensus}}\label{sec:resilient-strategy}

\subsection{\review{The Friedkin-Johnsen Dynamics}}\label{sec:FJ-dynamics}

	Because the classical consensus is fragile to misbehaving \nodes,
	we look for alternative strategies to minimize~\eqref{eq:cons-error-regular}.

	To this aim,
	\review{we step back to the optimization problem~\eqref{eq:network-cost-regular}
		and search for a way to make it more robust to unexpected behaviors.}
	In particular,
	\review{we modify the local problems associated with each regular \nodes $i\in\regSet$ 
		by integrating the nominal weight matrix $W^o$
		and adding a regularization term that penalizes deviations from the local observation:
	}
	%
	\begin{equation}\label{eq:utility-FJ}
		\review{\tilde{f}_{\text{local}}(x_i)} = \review{\lam}\left(x_i-\priornode{i}\right)^2 + (1-\review{\lam})\sumneigh{j}{i}W_{ij}^o\left(x_i-x_j\right)^2.
	\end{equation}
	\begin{ass}[\review{Nominal weights}]
		\review{The matrix $W^o$ is irreducible,
		row stochastic,
		and $W_{ii}^o=0, i\in\sensSet$ (no self-loops).}
	\end{ass}
	The parameter $ \review{\lam}\in[0,1] $ in~\eqref{eq:utility-FJ}
	makes
	the $ i $th \node anchor to its observation $\priornode{i}$,
	\review{so that large deviations of its state $x_i$ from $\priornode{i}$ are discouraged.}
	\review{We then let each agent greedily minimize the modified cost~\eqref{eq:utility-FJ}} at step $ k + 1 $, 
	which yields the
	celebrated Friedkin-Johnsen (FJ) dynamics~\cite{FJdynamics}:
	\begin{equation}\label{eq:FJ-dynamics}
		\xnode{i}{k+1} = \review{\lam}\priornode{i} + (1-\review{\lam})\sumneigh{j}{i}W_{ij}^o\xnode{j}{k}.\tag{FJ}
	\end{equation}
	\review{with $ \xnode{i}{0} = \priornode{i}$.
		We interpret the rule above as a modified consensus protocol
		where the \nodes do not fully align with neighbors
		but also \textit{compete} by tracking their own observation.
		In particular,
		we call the parameter $\lam$ as \textit{competition},
		referring to the case $ \lam = 0 $ (equivalent to the consensus protocol) as \emph{full collaboration}
		and to the case $ \lam = 1 $ as \emph{full competition}.
		}
	
	\review{While the dynamics~\eqref{eq:FJ-dynamics} is suboptimal if all agents are collaborative,
		because it prevents them from reaching a consensus if $\lam>0$,
		we use it to make regular \nodes resilient
		to unknown misbehaving \nodes. 
		Intuitively,
		anchoring a regular \node $i\in\regSet$ to its observation $\priornode{i}$ 
		prevents the \node from being arbitrarily dragged away by
		misleading values coming from misbehaving \nodes. 
		In particular,
		the latter \nodes obey~\eqref{eq:malicious-dynamics}
		with no relation to the protocol~\eqref{eq:FJ-dynamics} or nominal weights $W^o$.}
		
	In the following,
	we study \review{how} the \review{protocol}~\eqref{eq:FJ-dynamics} \review{improves system resilience.}
	In fact,
	tuning $ \lam $ within the interval $ [0,1] $
	originates a nontrivial \emph{\tradeoff}: 
	\review{
		what is the optimal competition $\lam$
		that makes regular \nodes most resilient with respect to task~\eqref{eq:network-cost-regular}?}
	\review{Exploring} this trade-off \review{under} misbehaving \nodes
	is the main matter of investigation of this article.
	\review{To this aim,
	we regard the consensus error 
	as function of the competition:
	this allows us to perform analysis and achieve insight about minimization of $\ereg(\lam)$.}

	\review{\begin{rem}[Connections with game theory and opinion dynamics]
			The FJ dynamics can be given the following game-theoretic interpretation.
		The cost~\eqref{eq:utility-FJ} is interpreted in games as \textit{cognitive dissonance},
		whereby a rational decision-maker gets incentive both in aligning with the neighbors and in following a local rule.
		Also,
		the function~\eqref{eq:utility-FJ} with $ \lam = 0 $ reduces to the utility used in~\cite{4814554}
		where the authors analyze the consensus protocol from a game-theoretic perspective.
		In opinion dynamics,
		the FJ dynamics is typically used to model \textit{prejudice},
		whereby the opinion of an agent is biased towards a personal belief despite interactions with others.
	\end{rem}	
	\begin{rem}[Competition for resilience]
		While most works in the literature regard $\lam$ as a model parameter,
		\textit{we purposely design $\lam$} in~\eqref{eq:FJ-dynamics}.
		The intuition behind this choice,
		seemingly counterintuitive for a collaborative task,
		is that introducing some competition among \nodes
		can mitigate behaviors that are unpredictable at design stage:
		rather than addressing the \textit{binary property} ``consensus is (not) achieved'' like typical works on resilient consensus,
		we take a broader viewpoint and interpret resilience as a \textit{real quantity}
		measured through the cost~\eqref{eq:cons-error-regular}.
	\end{rem}
}

	\begin{rem}[Heterogeneous \review{competition}]
		While we focus on a single parameter $ \lam $
		\review{for the sake of analysis,
			the general FJ model with a different parameter $\lami$ for each agent $i\in\sensSet$
		makes the analysis challenging	but does not affect the fundamental system behavior.}
		Designing a parameter $ \lami $ for each regular agent $i\in\regSet$
		to improve performance even further is an important topic,
		whose investigation is left to future work.
	\end{rem}

		\review{

\subsection{\review{Computation of the Consensus Error}}\label{sec:computation-consensus-error}

We now compute the error~\eqref{eq:cons-error-regular} with the steady state induced by
the dynamics~\eqref{eq:FJ-dynamics}. 
To this aim,
it is convenient to 
write the network dynamics associated with all regular \nodes.

First,
we highlight the interactions of regular and misbehaving \nodes
by partitioning the nominal weight matrix as
\begin{equation}\label{eq:nominal-W-partition}
	W^o = \left[\begin{array}{ c | c }
		\Wreg & \Wmal \\ 
		\hline
		* & *
	\end{array}\right]
	\ \Wreg \in\Real{R\times R}, \Wmal \in\Real{R\times M}.
\end{equation}
Then,
the dynamics of regular \nodes can be written as follows:
\begin{equation}\label{eq:system-dynamics}
	\begin{gathered}
		\xreg(k+1) = A\xreg(k) + B\review{\xnode{\malSet}{k}} + \lam\theta_\regSet \\
		A \doteq (1-\lam)\Wreg, \quad B \doteq (1-\lam)\review{\Wmal}.
	\end{gathered}
\end{equation}
If $ (1-\lam) \Wreg $ is Schur stable,
which happens if the graph is connected,
at steady state the dynamics~\eqref{eq:system-dynamics} induce the following distribution
w.r.t. the deception noises $\noise{k}$:
\begin{equation}\label{eq:FJ-dynamics-steady-state-components}
	\begin{gathered}
		\xregbar \doteq \lim_{k\rightarrow+\infty} \E[n]{\xreg(k)}, \quad
		P \doteq \lim_{k\rightarrow+\infty} \mathrm{Var}_n\left(\xreg(k)\right). 
	\end{gathered}
\end{equation}
Defining $ \selreg \doteq \left[ I_R \; | \; 0 \right] $,
the quantities above amount to
\begin{gather}
	\xregbar = \selreg L\left(\priorall+\selmal^\top v\right), \quad L \doteq \left(I-(1-\lam)W\right)^{-1}\lam\label{eq:L}\\
	P = (1-\lam)^2\Wreg P\Wreg^\top + (1-\lam)^2\Wreg Q\Wreg^\top\label{eq:P-lyapunov},
\end{gather}
where the matrix $W$ encodes the actual interactions (weights) followed within the network
and is defined as follows:
\begin{equation}\label{eq:actual-W-partition}
	W = \left[\begin{array}{ c | c }
		\Wreg & \Wmal \\ 
		\hline
		0 & I_M
	\end{array}\right].
\end{equation}
In particular,
$W$ means that the average state of a misbehaving \node is affected by no other \node,
according to~\eqref{eq:malicious-dynamics}.
The matrix $ L $ is row-stochastic with algebraic multiplicity of the eigenvalue $1$ equal to $M+1$
and does not induce a consensus.

Let 
$\E[x|y]{z} \doteq \E[x]{z|y}$
and let $\one\in\Real{R}$ denote the vector of all ones in $\Real{R}$.
From~\eqref{eq:cons-error-regular} and~\eqref{eq:system-dynamics}--\eqref{eq:P-lyapunov},
it follows:
\begin{equation}\label{eq:cons-error-explicit}
	\begin{aligned}
		\ereg 	&= \lim_{k\rightarrow+\infty}\Enorm[\priorall,v,n]{\xreg(k) - \one\thbarreg} \\
				&= \lim_{k\rightarrow+\infty}\E[\priorall,v]{\E[n]{\norm{\xreg(k) - \one\thbarreg}\big|\priorall,v}}\\
				&= \Enorm[\priorall,v]{\xregbar - \one\thbarreg} + \lim_{k\rightarrow+\infty}\mathrm{var}_n\left(\xreg(k)-\one\thbarreg\right).
	\end{aligned}
\end{equation}
Then,
standard calculations allow us to rewrite the consensus error as follows:
\begin{equation}\label{eq:cons-error-explicit-1}
	\ereg(\lam) = \eregv(\lam) + \eregn(\lam) + \kappa,
\end{equation}
where $\kappa$ does not depend on $\lambda$
and
\begin{equation}\label{eq:cons-error-explicit-2}
	\eregv(\lam) \doteq \tr{\Var E^\top E}, \quad 
	\eregn(\lam) \doteq \tr{P},
\end{equation}
where we define
$ \Var = \Varn + \selmal\selmal^\top V$,
$E=\selreg L - \Creg\selreg$,
and
$ \Creg\doteq\frac{1}{R}\one\one^\top $.
The expression~\eqref{eq:cons-error-explicit-1}
highlights that the two features of the misbehavior modeled in~\cref{ass:mal-node-dynamics}
generate two different contributions to the consensus error.
The error term $\eregv$ is caused by the biased observations of misbehaving \nodes $\priornode{} + v$
that are constantly injected into the dynamics~\eqref{eq:system-dynamics}. 
Instead,
the error term $\eregn$ is produced by the deception noises $\noise{k}$ that 
make the steady state drift away.

\begin{lemma}[Drift \vs competition]\label{lem:P-monotonic}
	The error term $\eregn(\lam)$ is strictly decreasing with $\lam$
	and $\eregn(1)=0$.
\end{lemma}
\begin{proof}
	See~\cref{app:proof-P}.
\end{proof}

In words,
\cref{lem:P-monotonic} implies that
setting $\lambda>0$ makes regular \nodes more resilient to the deception noise 
as opposed to the standard consensus protocol. 
This observation relates to~\cite{comoRobustness2015tnse}
where the authors observe that even small perturbations of a row-stochastic matrix $W$
can result in large norm of the matrix difference and
change of the Perron-Frobenius eigenvector.
}

\subsection{\review{The Competition-Collaboration Trade-off}}\label{sec:opt-lam-nontrivial}

\review{To study how our proposed approach performs in the presence of misbehaving \nodes,
	we first confront the two extreme cases of full collaboration and full competition
	to see when the former approach should be ruled out by default.}

\begin{prop}[Full competition \vs full collaboration]\label{prop:cons-err-mal}
	In the presence of misbehaving \nodes,
	the dynamics~\eqref{eq:FJ-dynamics} with $ \lam = 1 $
	yields a smaller error than with $ \lam=0 $
	if and only if
	\begin{multline}\label{eq:full-comp-better-full-coll-condition}
		\review{\dfrac{M^2}{R}\eregn(0) + \tr{V}} \ge  \dfrac{M^2}{R}\tr{\Varn_{11}} - \dfrac{2M^2}{R^2}\summat{\Varn_{11}} \\
		+ \dfrac{2M}{R}\summat{\Varn_{12}} -\summat{\Varn_{22}},
	\end{multline}
	where 
	$ \summat{A} = \sum_{i,j} A_{ij} $.
\end{prop}
\begin{proof}\extended{}{[Sketch of proof]}
	The statements follow from manipulations of the consensus errors
	induced by the two considered instantiations of~\eqref{eq:FJ-dynamics}.
	The full derivation is reported in~\cref{app:proof-full-comp-vs-full-coll}.
\end{proof}

In words,
\cref{prop:cons-err-mal} implies that
the fully competitive approach 
outperforms
the consensus protocol 
as soon as the misbehavior disturbances
are sufficiently intense
compared to the prior
correlations between regular and misbehaving \nodes. 

After acknowledging that the proposed competition-based approach can be
more resilient than the standard consensus protocol \review{in the presence of misbehaving \nodes},
we now turn to \review{study} the optimal resilient strategy.
\review{In other words, we are interested in choosing $ \lam $ so as to reduce the consensus error.}
In particular, we address the optimal \review{competition} $\lam^*$:
\begin{equation}\label{eq:optimal-lam}
	\lam^* \in \argmin\ereg(\lam).
\end{equation}
\review{Such an optimal parameter exists
by Weierstrass theorem because $ \ereg(\lam)$ is continuous in $(0,1]$ and
has a continuous extension at $ \lam = 0 $
through the extended continuity of $L$~\cite{7577815}.
}

The next \review{result describes when the optimal competition is nontrivial},
meaning that the regular \nodes should \review{compete against their neighbors} in order to minimize the error~\eqref{eq:cons-error-regular}.

\review{
	\begin{thm}[Competition-collaboration trade-off.]\label{prop:opt-lam-nontrivial}
		Let
		$\Gamma\doteq\lim_{\lam\rightarrow0^+}\frac{dL}{d\lam}$ with block partition
		\begin{equation}\label{eq:Gamma-partition}
			\Gamma = \left[\begin{array}{ c | c }
				\Gamma_{1} & \Gamma_{2} \\ 
				\hline
				0 & 0
			\end{array}\right], \quad \Gamma_{1} \in \Real{R\times R}, \ \Gamma_{2} \in \Real{R\times M}
		\end{equation}
		and $ C_{RM} \doteq \frac{\one_R\one_M^\top}{M} $.
		If either of the following conditions holds:
		\begin{enumerate}[label=C\arabic*.]
			\item $\Varn$ is diagonal;
			\item $W^o$ is symmetric
				and
				\begin{multline}\label{eq:opt-lam-greater-zero-condition}
					\review{-\dfrac{\d{\eregn}(0)}{\d{\lam}}} - \tr{V\Gamma_{2}^\top C_{RM}} > \tr{-\Varn_{11}\Gamma_{1}^\top\Creg-\right.\\
						\left.\Varn_{12}\Gamma_{2}^\top\Creg+\Varn_{12}^\top\Gamma_{1}^\top C_{RM}+\Varn_{22}\Gamma_{2}^\top C_{RM}};
				\end{multline}
		\end{enumerate}
		then $ \lam^*\in (0,1) $.
	\end{thm}}
\begin{proof}[Sketch of proof]
	The result is proven in two phases.
	Firstly, 
	we show that $ \lam^* < 1 $:
	we compute the first derivative of $ \ereg(\lam) $ at $ \lam = 1 $
	and show that it is positive,
	hence $\ereg(\lam)$ is strictly increasing in a left neighborhood of $ \lam = 1 $.
	Secondly,
	we show that $ \lam^* > 0 $:
	we compute the right derivative of $\ereg(\lam)$ as $ \lam \rightarrow 0^+$
	and show that it is negative,
	hence the error function is strictly decreasing in a right neighborhood of $ \lam = 0 $.
	The detailed calculations are reported in~\cref{app:proof-opt-lam-nontrivial}.
\end{proof}

Intuitively,
\review{any optimal parameter $\lam^*$
is strictly between $0$ and $1$ if the misbehavior is sufficiently
disruptive so that the consensus protocol yields poor performance,
similarly to what remarked below~\cref{prop:cons-err-mal},
while full competition is never optimal under our standing assumptions.}

\begin{rem}[Optimal \review{competition} with general matrices]
	\review{Even though we assume regular agents have no self-loops, 
		\cref{prop:opt-lam-nontrivial} holds also if this is relaxed.
		Further, 
		we numerically show that $\lam^*\in(0,1)$
		if $ W^o $ is row stochastic and $\Varn$ is not diagonal.
	}

\begin{rem}[Optimal \review{competition} with zero noise]
	\cref{prop:opt-lam-nontrivial} implies that $ \lam^* $
	may be positive 
	\review{even if $V$ and $Q$ are zero.}
	\review{This is indeed consistent with the misbehavior model:
	not only misbehaving \nodes corrupt the consensus value through deception bias and deception noise 
	but mostly they behave against the prescribed protocol,
	so that full collaboration is in general a poor strategy even if $v$ and $\noise{k}$ are trivial.}
\end{rem}
\end{rem}

\subsection{Performance \vs \xspace\review{Misbehavior}}\label{sec:opt-lam-vs-d}

We now study how the performance of the dynamics~\eqref{eq:FJ-dynamics} varies
\review{with deception biases $ v $ and deception noises $\noise{k}$.}

We first show an intuitive result:
more disruptive misbehavior induce larger consensus errors
for every $ \lam $.

\begin{prop}[\revision{Performance \vs misbehavior}]\label{prop:error-increases-with-d}
	\revision{The error $ \ereg $} is strictly increasing with \review{$ V $} 
	\review{and with $Q$ w.r.t. the partial order of semi-definite matrices.}
\end{prop}
\begin{proof}
	See~\cref{app:proof-error-vs-d}.
\end{proof}

\revision{We next \review{study} what happens to the optimal \review{competition} $ \lam^* $.}
Intuitively, the more the nominal system behavior is disrupted,
the more regular agents should benefit from competing rather than
collaborating with (potential) misbehaving neighbors.
Formally speaking, 
this requires $ \lam^* $ to increase with
the intensities of \review{deception biases and noises}. 
Such a claim is hard to prove analytically because of the \review{involved} structure of the cost function. 
In particular, studying the second derivative of $\ereg(\lam)$ is complicated by the
\review{asymmetric matrix inside the trace of $\eregv(\lam)$},
and similarly,
a unique root of the first derivative of $\ereg(\lam)$ cannot be proved, in general.

\review{Nonetheless,
	the next results contribute towards our intuition
	by describing how the minimum points vary with the misbehavior.}
\review{For convenience,
	we denote the diagonal elements of the covariance matrices by $d_m\doteq V_{mm}$
	and $q_m\doteq Q_{mm}$.}

\begin{prop}[\revision{Optimal \review{competition} \vs misbehavior}]\label{prop:opt-lam-vs-noise}
	\review{Let $ \lam_{\text{min}} $ \review{be a minimum point of $\ereg(\lam)$},
		then $ \lam_{\text{min}} $ is strictly increasing with $ d_m, m\in\malSet $,}
	\review{and with $Q$ w.r.t. the partial order of semi-definite matrices.}
\end{prop}
\begin{proof}
	See~\cref{app:proof-opt-lam-vs-noise}.
\end{proof}


An immediate consequence of~\cref{prop:opt-lam-vs-noise} is that,
if there is a unique minimum point
for some values of $ V $ and $Q$,
then there is a unique minimum point for any ``larger'' $ V $ and $Q$,
which corresponds to $ \lam^* $.
\revision{In words, a more \review{disruptive misbehavior}
force regular \nodes to progressively become more competitive,
in order not to be deceived by misbehaving \nodes that can draw them away from the nominal average consensus.}
The next proposition refines to this result by describing \review{the optimal parameter $\lam^*$}
\revision{with ``extreme'' \review{misbehavior}}.

\begin{prop}[\revision{Optimal \review{competition} with extreme \review{misbehavior}}]\label{prop:opt-lam-noise-inf}
	\review{Let $ \lam_{\text{min}} $ be a minimum point of $\ereg(\lam)$,
		then $ \lim_{d_m\rightarrow\infty}\lam_{\text{cr}}(d_{m}) = 1 $
		and $ \lim_{q_{m}\rightarrow\infty}\lam_{\text{cr}}(q_{m}) = 1, m\in\malSet $.}
\end{prop}
\begin{proof}
	See~\cref{app:proof-opt-lam-noise-inf}.
\end{proof}

According to intuition,
the (trivial) optimal strategy for regular \nodes
is to fully compete when
\revision{the \review{misbehavior} is too \review{disruptive}.}
However, numerical tests in the next section
show that $ \lam^* $ is significantly smaller than $ 1 $ in several cases.


\section{Numerical Experiments}\label{sec:numerical-tests}

\begin{figure}
	\centering
	\begin{subfigure}{0.5\linewidth}
		\centering
		\includegraphics[width=.99\linewidth]{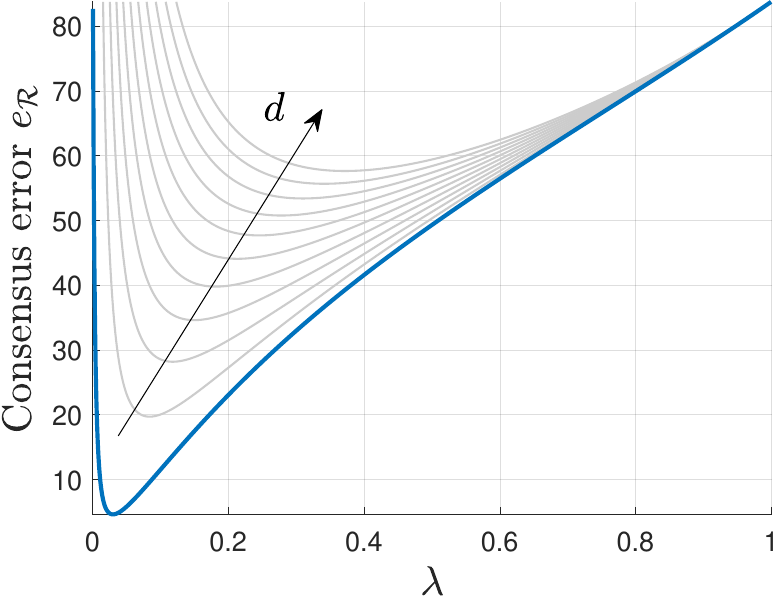}
		\caption{Average consensus error~\eqref{eq:cons-error-regular}}
		\label{fig:err_FJ_reg(3)_d10-100_1mal_Varn_10^(-.2)_curve}
	\end{subfigure}%
	\begin{subfigure}{0.5\linewidth}
		\centering
		\includegraphics[width=.99\linewidth]{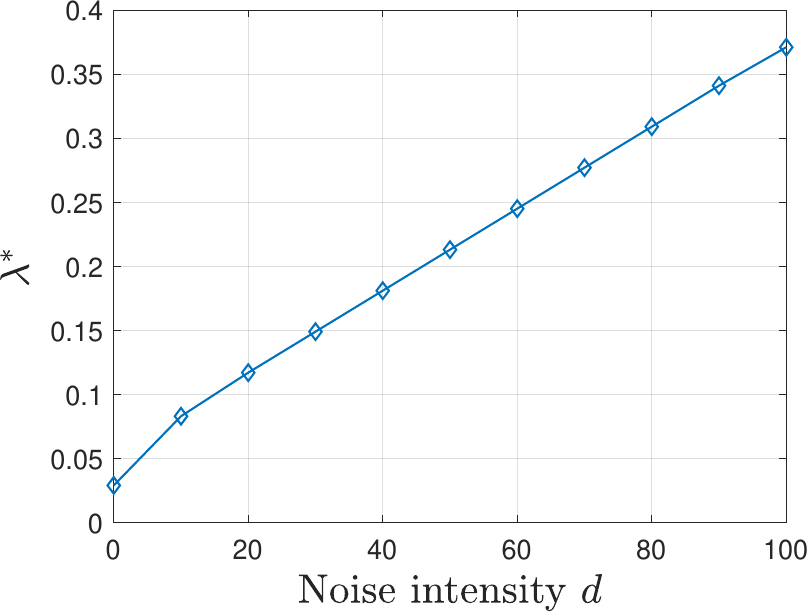}
		\caption{Optimal $ \lam $ as a function of $ d $.}
		\label{fig:err_FJ_reg(3)_d10-100_1mal_Varn_10^(-.2)_opt_lam}
	\end{subfigure}
	\caption{FJ dynamics consensus error with $ 3 $-regular graph, exponential decay of observation covariances, and one misbehaving \node.
		The arrow shows how the error curve varies as the intensity $ d $ of the \review{deception bias} increases.
	}
	\label{fig:err_FJ_reg(3)_d10-100_1mal_Varn_10^(-.2)}
\end{figure}

In this section,
we perform numerical experiments
on the consensus error $ \ereg $ 
to achieve intuition about the behavior of FJ dynamics under different topologies and \review{misbehavior}, 
and draw insight about effective choices of the parameter $ \lam $.

In~\autoref{fig:err_FJ_reg(3)_d10-100_1mal_Varn_10^(-.2)}, 
we considered a $ 3 $-regular communication graph with $ 100 $ \nodes
and uniform weights\footnote{
	\review{In a $k$-regular graph,
	each node has exactly $k$ neighbors.}
}. 
The prior covariance $ \Varn $ was chosen such that,
for each \node $ i $,
the cross-covariances obeyed an exponential decay,
$ \varnode[i]{j} = 10^{-0.2\ell(i,j)} $,
$ \mathrm\ell(i,j) $ being the length of a shortest path between $ i $ and $ j $, 
with $ \varnode{i} \equiv 1 $.
Further, we randomly selected one misbehaving \node and
varied the intensity of its deception bias $ d $ within the range $ [0, 100] $,
\review{with constant intensity of deception noise $q$.}

Figure~\ref{fig:err_FJ_reg(3)_d10-100_1mal_Varn_10^(-.2)_curve} shows the error curve
as $ d $ increases.
All curves exhibit a unique minimum point $ \lam^* $, plotted in~\autoref{fig:err_FJ_reg(3)_d10-100_1mal_Varn_10^(-.2)_opt_lam}.
Further, both error curve and minimum point 
increase with $ d $,
according to~\cref{prop:error-increases-with-d,prop:opt-lam-vs-noise}, 
showing that the competition level needs to grow with the intensity of deception biases.
\review{The same qualitative behavior was observed by varying $q$.}

\begin{figure}
	\centering
	\begin{subfigure}{0.5\linewidth}
		\centering
		\includegraphics[width=.99\linewidth]{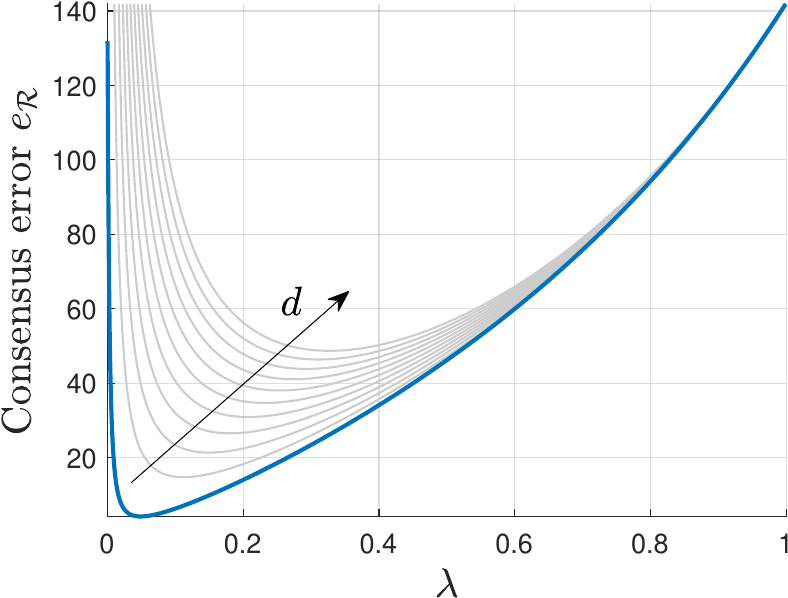}
		\caption{Average consensus error~\eqref{eq:cons-error-regular}}
		\label{fig:err_FJ_reg(3)_d10-100_1mal_Varn_diag_rand[1,2]_curve}
	\end{subfigure}%
	\begin{subfigure}{0.5\linewidth}
		\centering
		\includegraphics[width=.99\linewidth]{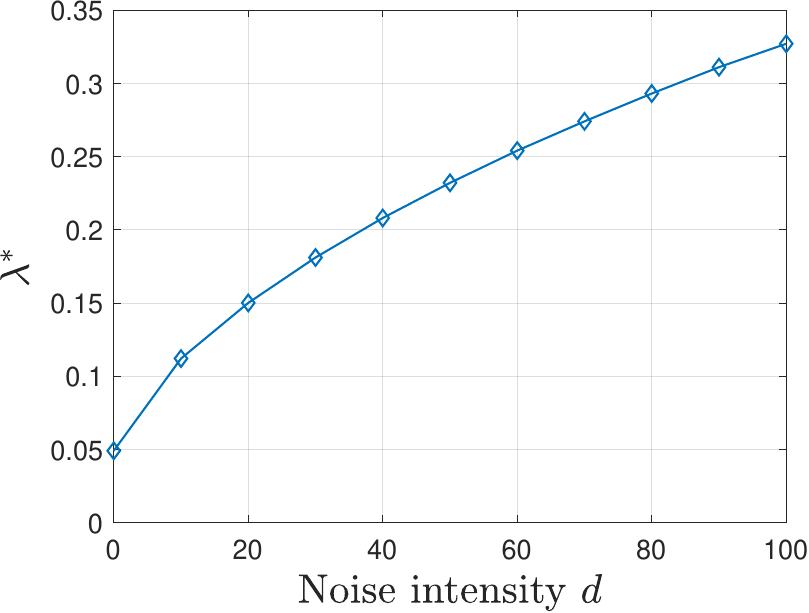}
		\caption{Optimal $ \lam $ as a function of $ d $.}
		\label{fig:err_FJ_reg(3)_d10-100_1mal_Varn_diag_rand[1,2]_opt_lam}
	\end{subfigure}
	\caption{FJ dynamics consensus error with $ 3 $-regular graph, diagonal prior covariance matrix $\Varn$, and one misbehaving \node.
	}
	\label{fig:err_FJ_reg(3)_d10-100_1mal_Varn_diag_rand[1,2]}
\end{figure}

Figure~\ref{fig:err_FJ_reg(3)_d10-100_1mal_Varn_diag_rand[1,2]} shows the same experiment
but with a diagonal covariance matrix $ \Varn $.
We observe the same monotonic behavior of $ \ereg $ and $ \lam^* $.
Further, we note that the error curve has a convex shape.
In fact, even though it was not possible to prove it formally,
all tests performed with diagonal covariance matrices resulted in strictly convex error functions from numerical tests.%


\begin{figure}
	\centering
	\begin{subfigure}{0.5\linewidth}
		\centering
		\includegraphics[width=.99\linewidth]{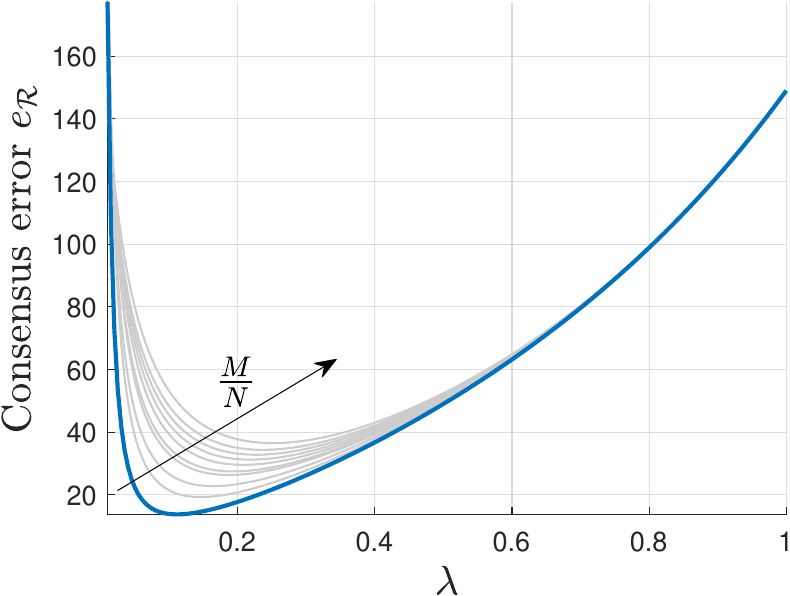}
		\caption{Average consensus error~\eqref{eq:cons-error-regular}}
		\label{fig:err_FJ_reg(3)_d10_mal1-10_reg100_Varn_diag_rand[1,2]_curve}
	\end{subfigure}%
	\begin{subfigure}{0.5\linewidth}
		\centering
		\includegraphics[width=.99\linewidth]{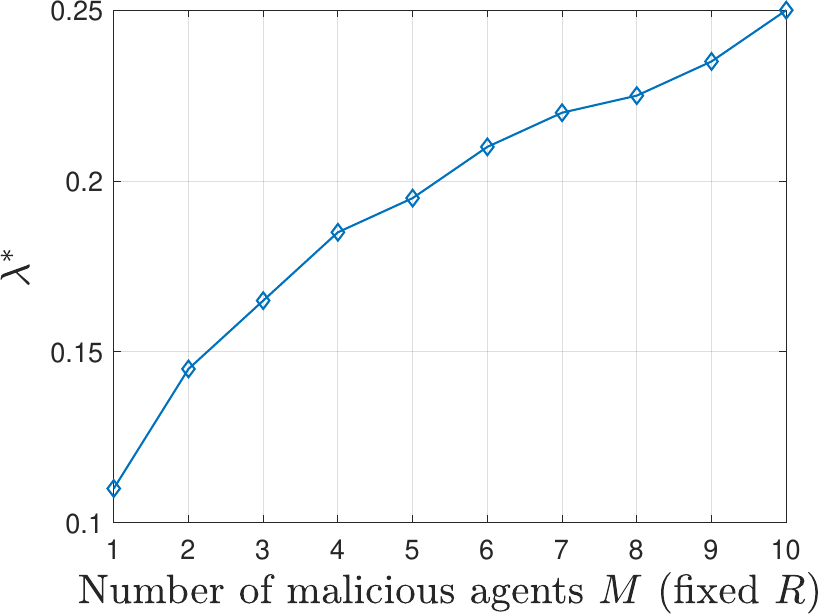}
		\caption{Optimal $ \lam $ as a function of $ M $.}
		\label{fig:err_FJ_reg(3)_d10_mal1-10_reg100_Varn_diag_rand[1,2]_opt_lam}
	\end{subfigure}
	\caption{FJ dynamics consensus error with $ 3 $-regular graph and diagonal prior covariance matrix $\Varn$.
		The arrow on the left box shows how the error varies as the number of misbehaving nodes $ M $ increases (with $ R = 100 $).
	}
	\label{fig:err_FJ_reg(3)_d10_mal1-10_reg100_Varn_diag_rand[1,2]}
\end{figure}

\begin{figure}
	\centering
	\begin{subfigure}{0.5\linewidth}
		\centering
		\includegraphics[width=.99\linewidth]{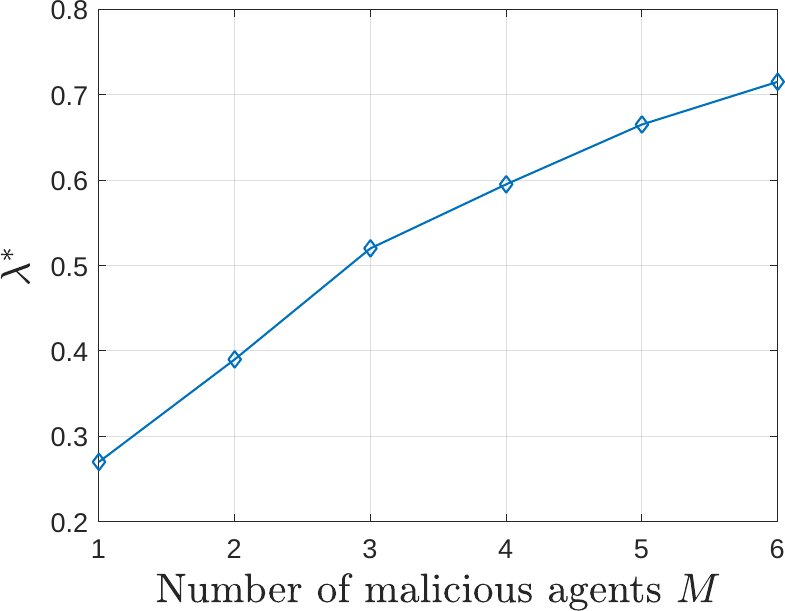}
		\caption{Exponential observation covariances.}
		\label{fig:err_FJ_reg(3)_d10_mal1-6_Varn_10^(-.2)_lam_opt}
	\end{subfigure}%
	\begin{subfigure}{0.5\linewidth}
		\centering
		\includegraphics[width=.99\linewidth]{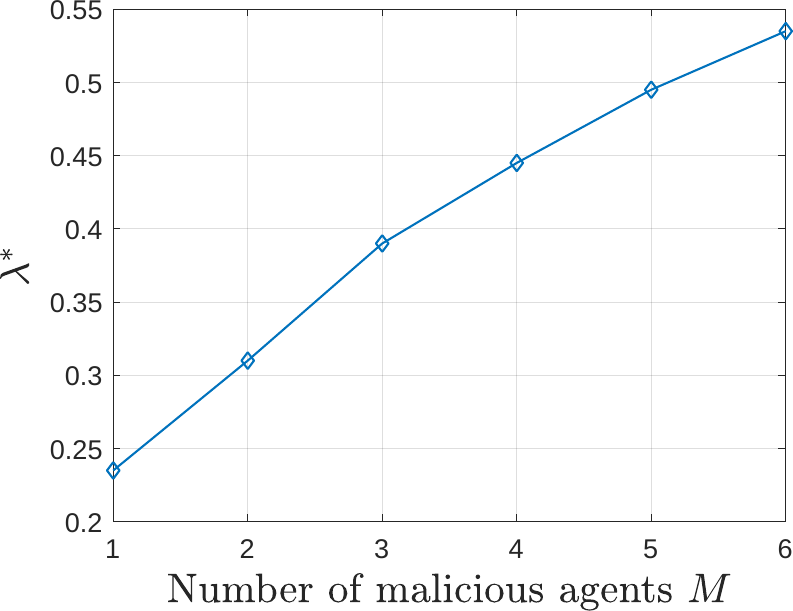}
		\caption{Diagonal observation covariance matrix.}
		\label{fig:err_FJ_reg(3)_d10_mal1-6_Varn_diag_rand[1,2]_lam_opt}		
	\end{subfigure}
	\caption{Optimal $ \lam $ as a function of $ M $ with $ d = 10 $.
		Each pair of misbehaving \nodes affects the same regular \node (\eg the first two belong to $ \neigh{1} $). 
	}
	\label{fig:err_FJ_reg(3)_d10_mal1-6}
\end{figure}

We next studied what happens when increasing the number of misbehaving \nodes $ M $.
To better visualize changes in the behavior of the system,
we fixed the set $ \regSet $ to be a network composed of $ R = 100 $ regular \nodes,
and added misbehaving \nodes across the network.
Figure~\ref{fig:err_FJ_reg(3)_d10_mal1-10_reg100_Varn_diag_rand[1,2]} shows
the error curve when $ 10 $ such \nodes are progressively introduced.
In particular, 
in this example,
all misbehaving \nodes are selected so as to affect different portions of the network,
which allows $ \lam^* $ to have relatively low values, see~\autoref{fig:err_FJ_reg(3)_d10_mal1-10_reg100_Varn_diag_rand[1,2]_opt_lam}.
Conversely, we note that, in the opposite scenario, 
some regular \nodes may be forced to almost freeze their observations (large $ \lam $) to not drive the error too large.
Figure~\ref{fig:err_FJ_reg(3)_d10_mal1-6} shows two cases where the misbehaving \nodes
are connected to the same regular \nodes. 
In particular, each couple is added to the neighborhood of one regular \node
(\eg the first two misbehaving \nodes added to the network are neighbors of \node $ 1 \in\regSet $).
In this case, $ \lam^* $ increases faster than~\autoref{fig:err_FJ_reg(3)_d10_mal1-10_reg100_Varn_diag_rand[1,2]_opt_lam},
because the regular \nodes affected by multiple misbehaving need to keep their error small:
in other words,
they can hardly collaborate because of their misbehaving neighbors.
We note that
$ \lam^* $ grows faster when observations of regular \nodes are correlated (\autoref{fig:err_FJ_reg(3)_d10_mal1-6_Varn_10^(-.2)_lam_opt}),
because such \nodes can trust that their states may be similar even before starting dynamical updates,
and competing is less risky than collaborating.

\begin{figure}
	\centering
	\begin{subfigure}{0.5\linewidth}
		\centering
		\includegraphics[width=.99\linewidth]{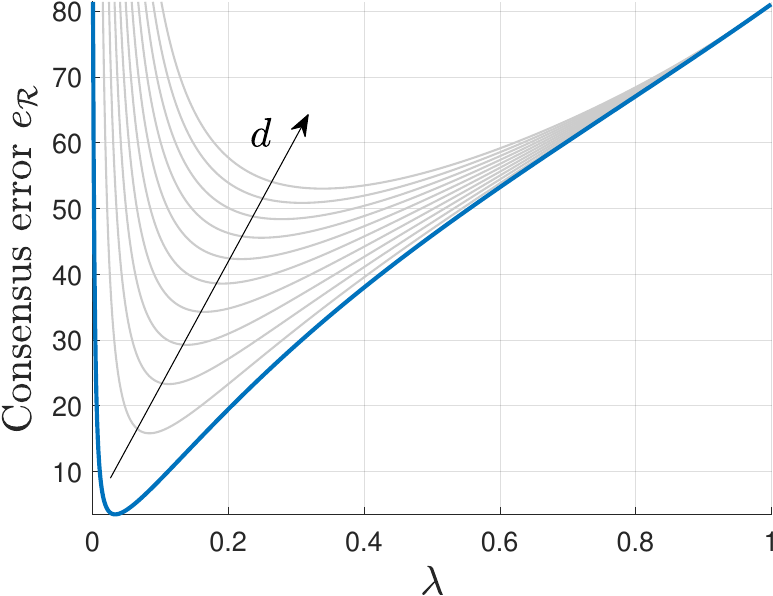}
		\caption{Average consensus error~\eqref{eq:cons-error-regular}}
		\label{fig:err_FJ_reg(3,4)_d10_mal1-4_Varn_10^(-.2)_curve}
	\end{subfigure}%
	\begin{subfigure}{0.5\linewidth}
		\centering
		\includegraphics[width=.99\linewidth]{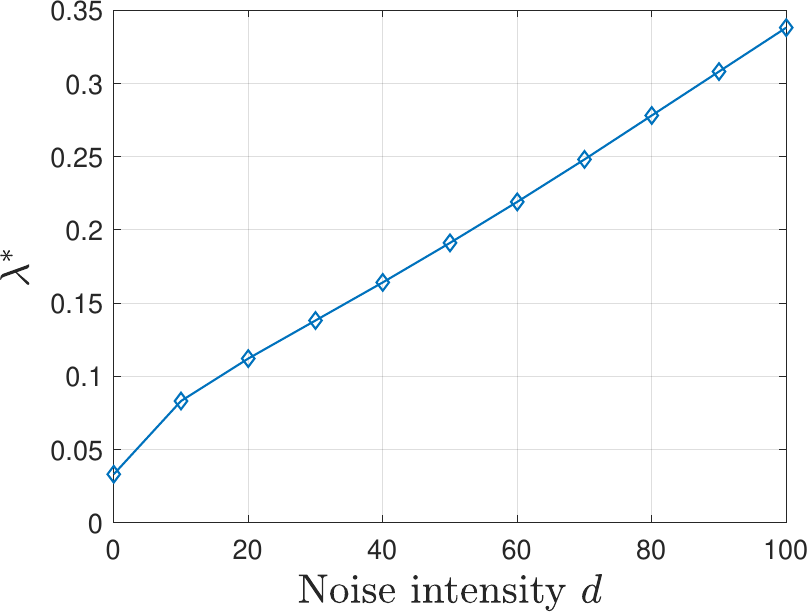}
		\caption{Optimal $ \lam $ as a function of $ d $.}
		\label{fig:err_FJ_reg(3,4)_d10_mal1-4_Varn_10^(-.2)_lam_opt}
	\end{subfigure}
	\caption{FJ dynamics consensus error with $ (3,4) $-degree communication graph,
		exponential decay of observation covariances, and one misbehaving \node.
		}
	\label{fig:err_FJ_reg(3,4)_d10_mal1-4_Varn_10^(-.2)}
\end{figure}

Finally, it is interesting to see that the error behavior observed above 
is consistent also if $ \review{W^o} $ is only row-stochastic,
thus yielding nonzero consensus error even \review{in the nominal scenario}. 
Figure~\ref{fig:err_FJ_reg(3,4)_d10_mal1-4_Varn_10^(-.2)} shows consensus error and $ \lam^* $
when each node in the graph has degree $ 3 $ or $ 4 $
and $ \review{W^o} $ has uniform weights.

Other numerical tests performed with different graphs,
observation distributions, and choice of the misbehaving \nodes
show the same quasi-convex behavior of the error function
and are omitted in the interest of space.
This reinforces and extends the scope of our formal analysis,
showing that indeed the \tradeoff
emerges as a natural resilient mechanism for multi-agent systems.


\begin{rem}[Value of optimal $ \lam $]
	A remarkable feature of the FJ dynamics
	that emerges from the tests above
	is that $ \lam^* $ is usually small
	(within the interval $ [0.1,0.2] $ in many cases).
	This translates into the practical advantage that
	adding a little competition may be sufficient to 
	get 
	a good level of resilience
	without forcing too conservative updates by regular \nodes.
\end{rem}

\subsection{\titlecap{competition-collaboration trade}-off: \titlecap{analytical insight}}\label{sec:trade-off}

As mentioned earlier, 
the consensus error function $ \ereg(\lam) $ 
is hard to study and an exhaustive analysis seems not possible.


Some intuition can be achieved from a decomposition 
that we study next.
To keep notation light, we assume a single misbehaving \node (with label $ m $)
and a diagonal covariance matrix $ \Varn $.
Then,
we can expand the consensus error as follows:
\begin{equation}\label{eq:cons-error-regular-split}
	\ereg = \underbrace{\sumreg{i}\varnode{i}\left\| L_i^{-m} - \dfrac{\one}{R}\right\|^2}_{\doteq e_{\regSet,\text{consensus}}} \: + \:
	\underbrace{\left(\varnode{m}+d\right)\left\| L_m^{-m}\right\|^2 \review{+ \eregn}}_{\doteq e_{\regSet,\text{deception}}}.
\end{equation}
In~\eqref{eq:cons-error-regular-split},
$ L_i\in\Real{N} $ is the $ i $th column of $ L $
and $ L_i^{-m}\in\Real{N-1} $ is obtained from $ L_i $ by removing its $ m $th row
(\revision{corresponding to the misbehaving \node}).
The error curves 
are shown in~\autoref{fig:trade-off}.
Equation~\eqref{eq:cons-error-regular-split} allows for an intuitive interpretation of the error,
which leverages the notion of \textit{social power}~\cite{cartwright1959studies,9327497}.

In opinion dynamics, the social power is used to quantify how much
the opinion of an agent affects the opinions of all agents.
In particular, when opinions evolve according to the FJ dynamics,
the element $ L_{ij} $ quantifies the influence of
agent $ j $ on agent $ i $:
as $ L_{ij} $ increases,
agent $ i $ is more affected by the initial opinion of agent $ j $.
The total social power of agent $ j $ is a symmetric and increasing function of all elements $ \{L_{ij}\}_{i\in\sensSet} $.%
\footnote{References~\cite{cartwright1959studies,9327497} use the arithmetic mean of $ \{L_{ij}\}_{i\in\sensSet} $.}

\begin{figure}
	\centering
	\includegraphics[width=.8\linewidth]{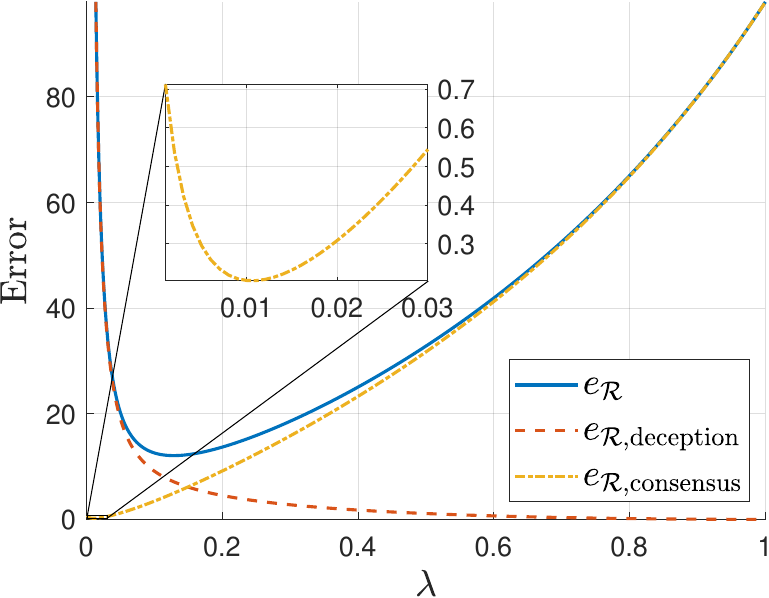}
	\caption{Consensus error and its two contributions in~\eqref{eq:cons-error-regular-split}.}
	\label{fig:trade-off}
\end{figure}

Borrowing such concepts from opinion dynamics
allows us to interpret the two contributions separated in~\eqref{eq:cons-error-regular-split}.
The first, $ e_{\mathcal{R},\text{deception}} $,
quantifies the impact of the misbehaving \node $m$ on regular \nodes.
The \textquotedblleft social power" of $ m $, 
as quantified through the vector $ L_{m}^{-m} $,
depends on the communication matrix $ W $ and on the parameter $ \lam $.
\revision{Each coordinate of $ L_{m}^{-m} $ decreases with $ \lam $,
meaning that the influence of the misbehaving \node weakens as regular \nodes anchor more tightly to their observations,
and becomes zero when $ \lam = 1 $,
namely, in the full-competition regime.
\review{We formalize this discussion as the following lemma.
	\begin{lemma}
		The component $e_{\regSet,\text{deception}}$ is decreasing with $\lam$.
	\end{lemma}
	\begin{proof}
		By computing the derivative of $L$ w.r.t. $\lam$,
		we see that each element of $L_{m}^{-m}$ is nonincreasing with $\lam$.
		Because $L$ is a nonnegative matrix,
		this and~\cref{lem:P-monotonic} yield the claim.
		See~\cref{app:L-derivative-limit} for the detailed calculations.
	\end{proof}
}}

The second contribution  $ e_{\mathcal{R},\text{consensus}} $
measures \textquotedblleft democracy" among regular \nodes,
\ie it is proportional to the mismatch between how much each regular \node affects the others
and the ideal value $ \nicefrac{1}{R} $,
which means that each \node affects all others equally.
This cost is zero if and only if the submatrix of $ L $
corresponding to interactions among regular \nodes is the consensus matrix: 
this can happen only if 
they do not interact with the misbehaving \node~\cite{7577815},
in which case
the vector $ L_m^{-m} $ is zero
(the misbehavior has no effect).
In this special case,
$ e_{\mathcal{R},\text{consensus}} $ is zero at $ \lam = 0 $
and increases monotonically as the network shifts
from a democratic system
where \nodes fully collaborate ($ \lam = 0 $)
to a disconnected system where \nodes fully compete ($ \lam = 1 $).
Conversely,
with misbehaving \nodes, 
$ e_{\mathcal{R},\text{consensus}} $ has a nontrivial minimizer (zoomed box in~\autoref{fig:trade-off}).
For small $ \lam $,
the misbehaving \node overrules all interactions
and regular \nodes hardly affect each other.
As $ \lam $ increases,
the interactions among regular \nodes become more relevant, 
making $ e_{\mathcal{R},\text{consensus}} $ decrease.
However, 
as $ \lam $ grows further,
the competition among regular \nodes becomes too aggressive
and makes them shift away from an ideal democratic system.

Overall, 
the error~\eqref{eq:cons-error-regular} has two concurrent causes that yield two regimes:
collaboration with misbehaving \nodes is most misleading for small $ \lam $,
while for large $ \lam $ the error is mainly due to regular \nodes that compete against each other 
and reject useful information shared by neighbors.
This matches intuition from~\eqref{eq:FJ-dynamics}
where $ \lam $ measures conservatism in \node updates.

\section{\titlecap{the role of the communication network}}\label{sec:network-optimization}

In the previous sections,
we discussed the benefits of using a competition-based approach (FJ dynamics)
to tame misbehaving \nodes. 
We now shift attention to the communication network,
in order to achieve intuition about resilient topologies. 
\revision{In~\autoref{sec:performance-metrics}, we introduce a second performance metric
which we use to evaluate resilience to attacks.
In~\autoref{sec:studying-network-topology},
we observe how performance varies with connectivity.}

\subsection{\revision{\titlecap{performance metrics}}}\label{sec:performance-metrics}
Besides consensus error,
we also aim to 
assess energy \review{spent to misbehave}. 
To this aim,
we interpret~\eqref{eq:system-dynamics} as a controlled system
where 
the misbehaving \nodes command the input $ \review{\xnode{\malSet}{\cdot}} $.
The controllability Gramian in $ K $ steps,
denoted by $ \Gram{K} $,
is defined for system~\eqref{eq:system-dynamics} as
\begin{equation}\label{eq:gramian}
	\Gram{K} = \sum_{k=0}^{K-1}A^kBB^\top (A^\top)^k.
\end{equation}
The controllability Gramian can be used to quantify the control effort:
the trace of $ \Gram{K} $,
called \emph{controllability index},
is inversely related to the control energy spent in $ K $ steps  (averaged over the reachable subspace),
as shown in literature~\cite{SRIGHAKOLLAPU2021105061,gu2015controllability,10.1093/comnet/cnz001}.
In words,
a small controllability index means that the misbehaviors
consume a lot of energy to steer $ \xreg $ across the reachable space, 
which may be desired to possibly drain out adversarial resources and hamper an external attack.

\review{If $M=1$,}
the controllability index can be written as
\begin{equation}\label{eq:trace-gramian}
	\tr{\Gram{K}} = (1-\lam)^2\sum_{k=0}^{K-1}\Big\|(1-\lam)^{k}\Wreg^k\review{\Wmal} \Big\|^2,
\end{equation}
resembling the consensus error component $ e_{\regSet,\text{deception}} $ in~\eqref{eq:cons-error-regular-split},
\begin{equation}\label{eq:cons-error-rewritten}
	e_{\regSet,\text{deception}} 
	\propto \Bigg\|\sum_{k=0}^{\infty}(1-\lam)^k\sum_{j=0}^{k-1}\Wreg^j\review{\Wmal} \Bigg\|^2.
\end{equation}
Both $ \tr{\Gram{K}} $ and $ e_{\regSet,\text{deception}} $ are decreasing with $ \lam $
(\ie the more competition, the better)
and depend on the vectors $ \Wreg^k\review{\Wmal} $
that describe how attacks spread in $ k $ steps.
The discount factor $ (1-\lam)^k $ makes the tail of the series in~\eqref{eq:cons-error-rewritten} negligible,
enhancing similarity between those two metrics.

\review{\begin{rem}[Controllability index]
	While we use~\cref{ass:mal-node-dynamics} to compute $\ereg$,
		the controllability Gramian in~\eqref{eq:gramian} is independent of the trajectory of the system
		and hence the controllability index evaluates an ``average trajectory'' of misbehaving \nodes.
\end{rem}}

\subsection{Network Connectivity \vs Resilience}\label{sec:studying-network-topology}


\review{We now explore how connectivity of the communication network
affects performance and resilience of the dynamics~\eqref{eq:FJ-dynamics}.
While in this sections we attempt to achieve heuristic intuition,
an analytical investigation is deferred to future work.
To this aim,
we fix the parameter $\lam = 0.1$
and numerically evaluate the theoretical performance as the density of the communication network increases.
Specifically,
for each evaluated network,
we assign uniform weights to the links
and compute consensus error $\ereg$ and controllability index $\tr{\Gram{K}}$
(where $K$ is the reachability index)
\review{selecting} \review{some \nodes as misbehaving} according to either of the following two cases:
\begin{itemize}
	\item the worst-case misbehaving \node,
	\ie $\malSet=\{m^*\}$ with
	\begin{align}
		&m^*=\argmax_{m\in\sensSet} \;\; \ereg, \label{eq:optimization-network-err}\\
		&m^*=\argmax_{m\in\sensSet} \;\; \tr{\Gram{K}}; \label{eq:optimization-network-contr}
	\end{align}
	\item five misbehaving \nodes randomly drawn from $\sensSet$.
\end{itemize}
We consider three common classes of graphs:
regular graphs with degree $\Delta$, 
Erd\"os-R\'enyi random graphs,
where a link between any two nodes exists with probability $p$,
and random geometric graphs,
where nodes are randomly placed in $[0,1]^2$
and any two nodes are linked if their distance is not greater than a radius $\rho$.
While regular graphs induce a doubly stochastic matrix even with simple uniform weights,
this is generally not true for the other graphs.
Hence,
to evaluate the consensus error $\ereg$,
we considered both the deviation from the nominal average defined in~\eqref{eq:cons-error-regular}
and the deviation from the consensus value computed from the left Perron eigenvector of the nominal weight matrix $W^o$.
Given that the results were qualitatively equal,
we report only the first case in the interest of space.

We consider networks with $N = 100$ \nodes
and compute the performance for each network
(\ie a combination of class of graph and density parameter)
by averaging over $1000$ random graphs for the worst-case misbehaving \node
and over $5000$ random graphs for the random selection of misbehaving \nodes.
The results are shown in~\cref{fig:graph-regular,fig:graph-erdos_renyi,fig:graph-geometric},
with the consensus error on the left and the controllability index on the right.}


\begin{figure}
	\centering
	\begin{minipage}[l]{.5\linewidth}
		\centering
		\includegraphics[width=\linewidth]{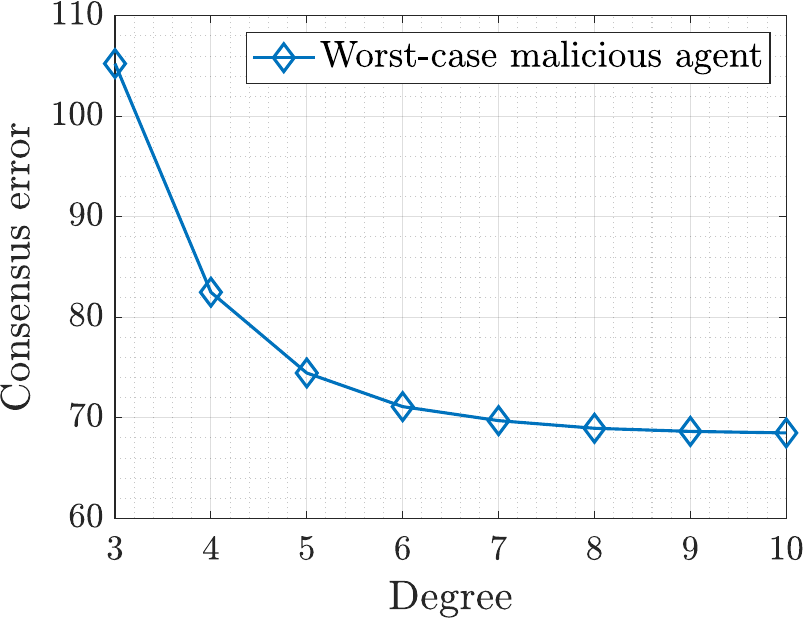}
	\end{minipage}%
	\begin{minipage}[r]{.5\linewidth}
		\centering
		\includegraphics[width=\linewidth]{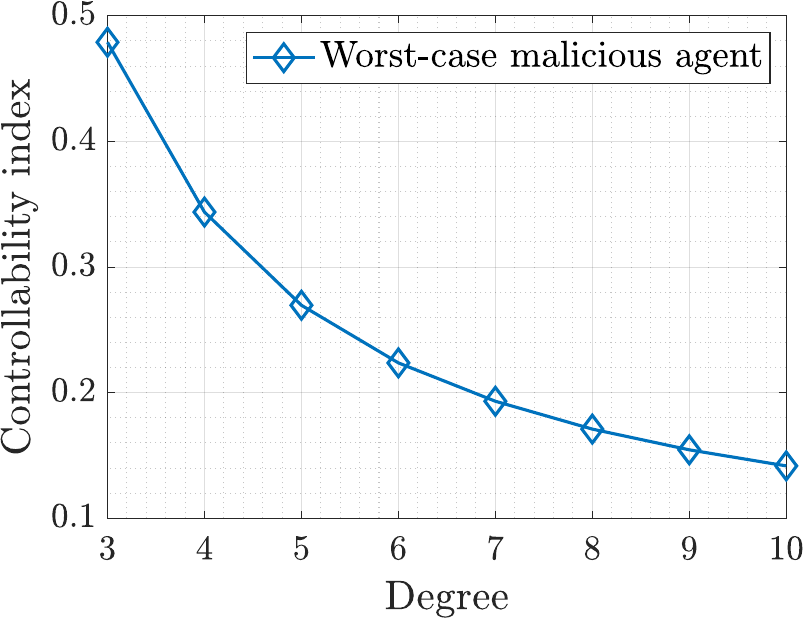}
	\end{minipage}\\
	\begin{minipage}[l]{.5\linewidth}
		\centering
		\includegraphics[width=\linewidth]{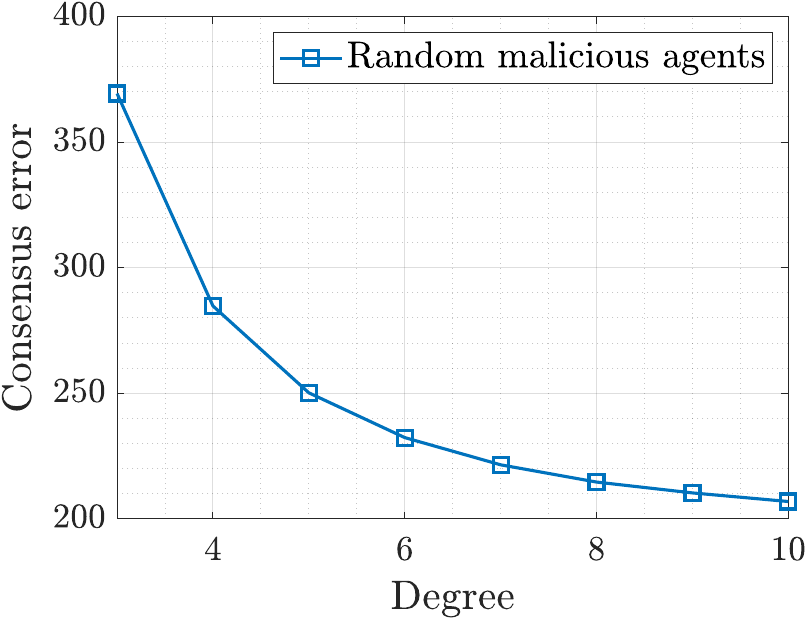}
	\end{minipage}%
	\begin{minipage}[r]{.5\linewidth}
		\centering
		\includegraphics[width=\linewidth]{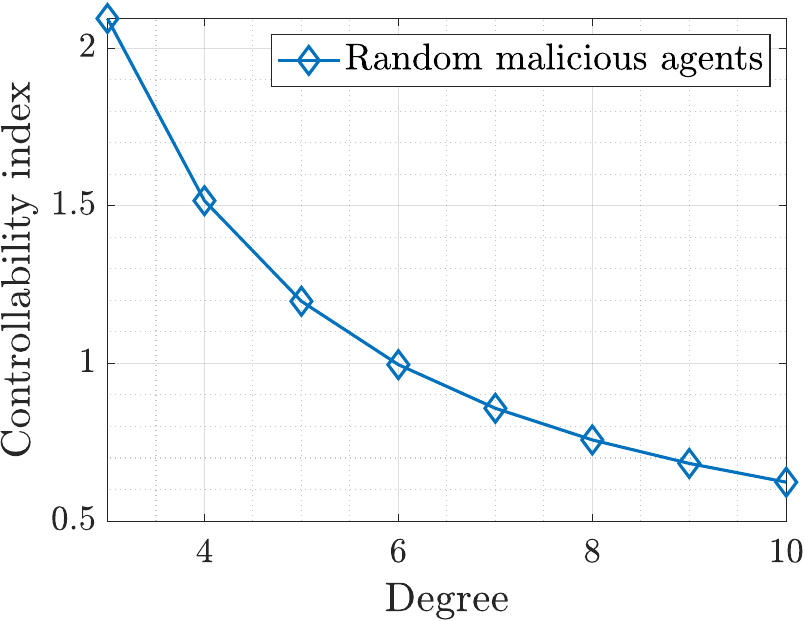}
	\end{minipage}
	\caption{\review{Average performance metrics for 
		regular graphs.} 
	}
	\label{fig:graph-regular}
\end{figure}

\begin{figure}
	\centering
	\begin{minipage}[l]{.5\linewidth}
		\centering
		\includegraphics[width=\linewidth]{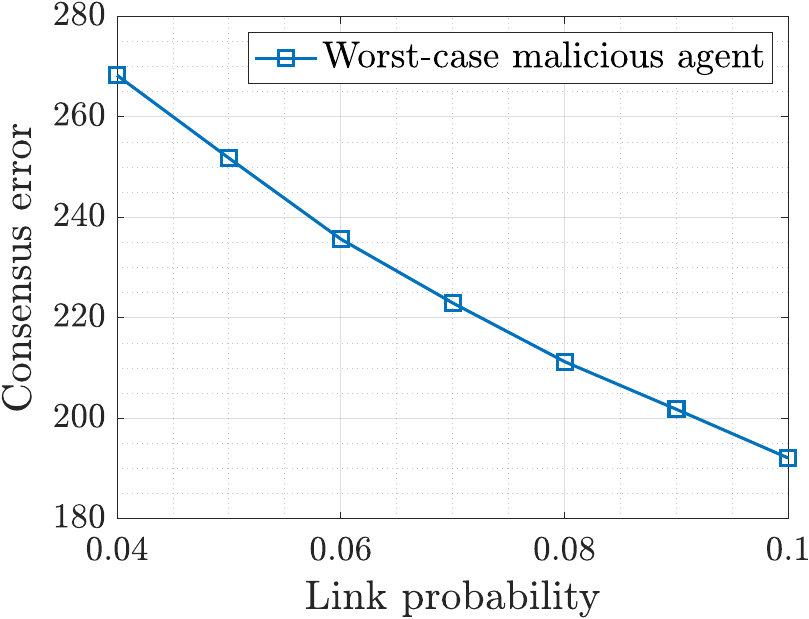}
	\end{minipage}%
	\begin{minipage}[r]{.5\linewidth}
		\centering
		\includegraphics[width=\linewidth]{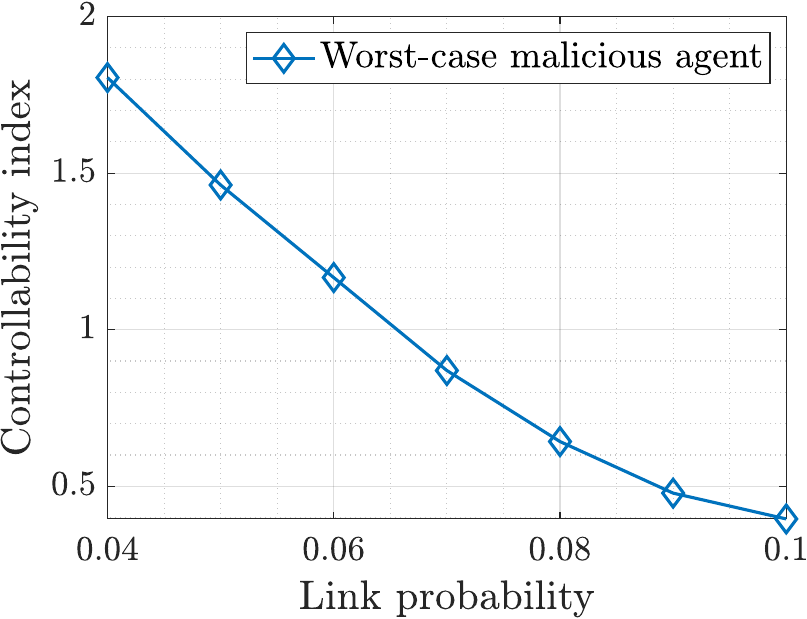}
	\end{minipage}\\
	\begin{minipage}[l]{.5\linewidth}
		\centering
		\includegraphics[width=\linewidth]{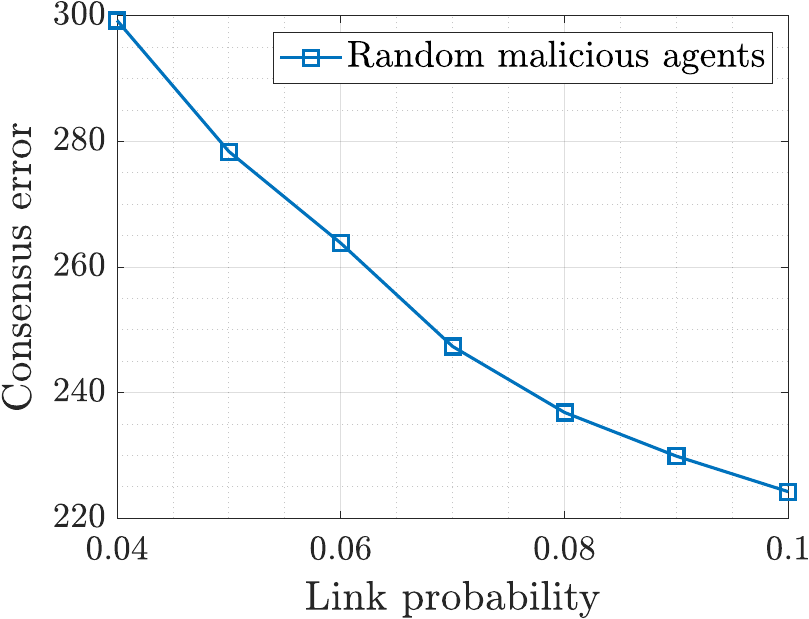}
	\end{minipage}%
	\begin{minipage}[r]{.5\linewidth}
		\centering
		\includegraphics[width=\linewidth]{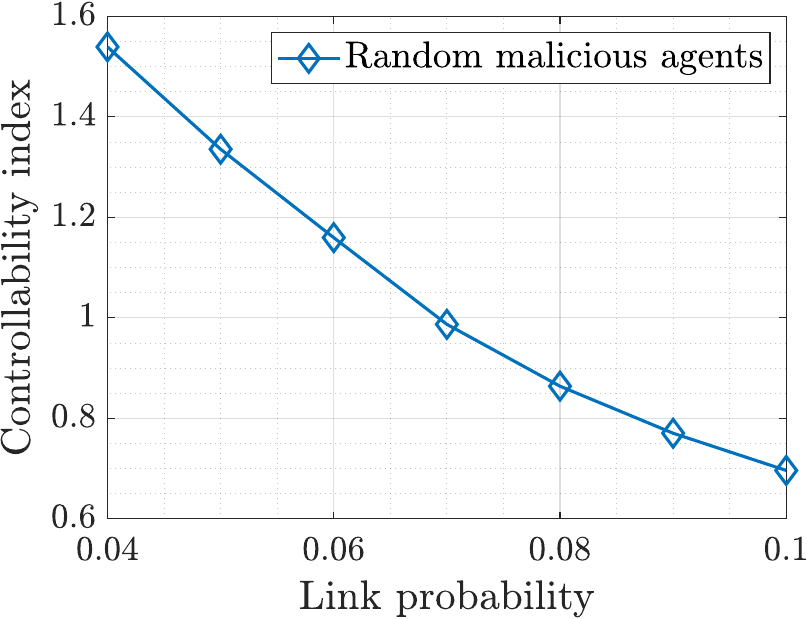}
	\end{minipage}
	\caption{\review{Average performance metrics for 
		Erd\"os-R\'enyi random graphs.} 
	}
	\label{fig:graph-erdos_renyi}
\end{figure}

\begin{figure}
	\centering
	\begin{minipage}[l]{.5\linewidth}
		\centering
		\includegraphics[width=\linewidth]{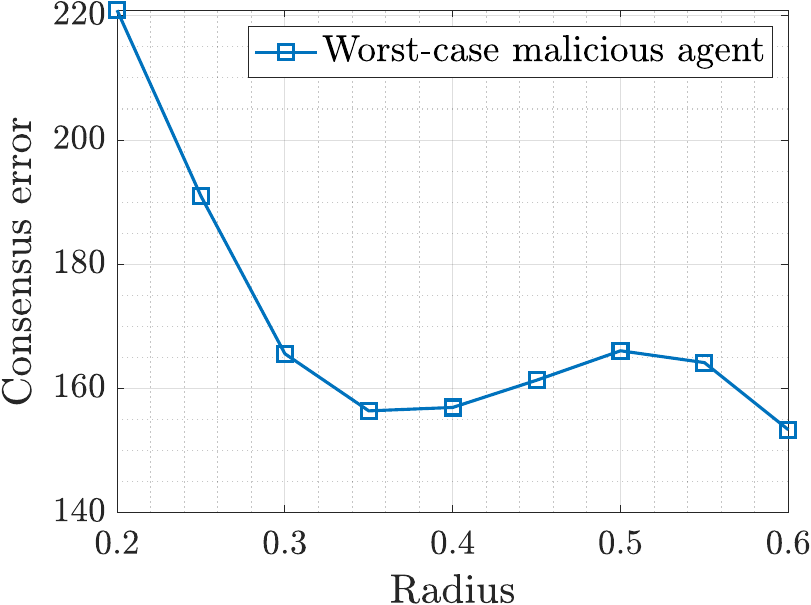}
	\end{minipage}%
	\begin{minipage}[r]{.5\linewidth}
		\centering
		\includegraphics[width=\linewidth]{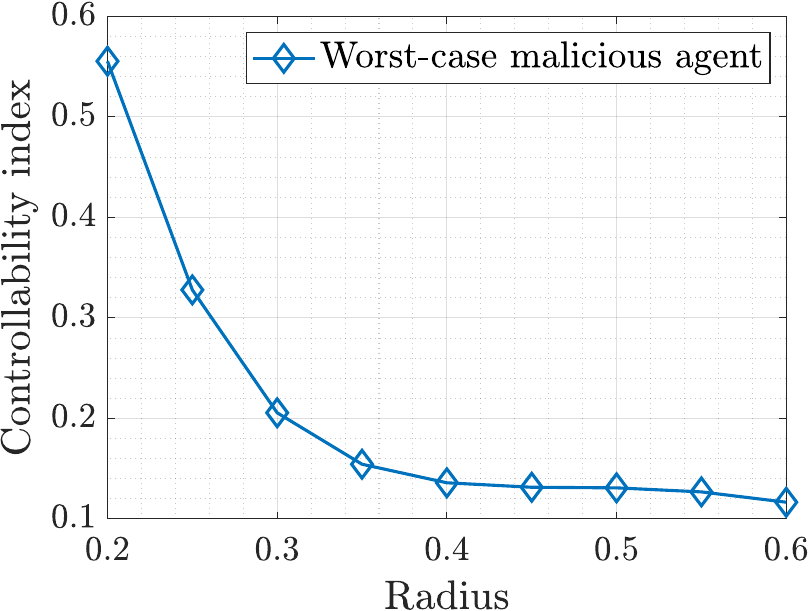}
	\end{minipage}\\
	\begin{minipage}[l]{.5\linewidth}
		\centering
		\includegraphics[width=\linewidth]{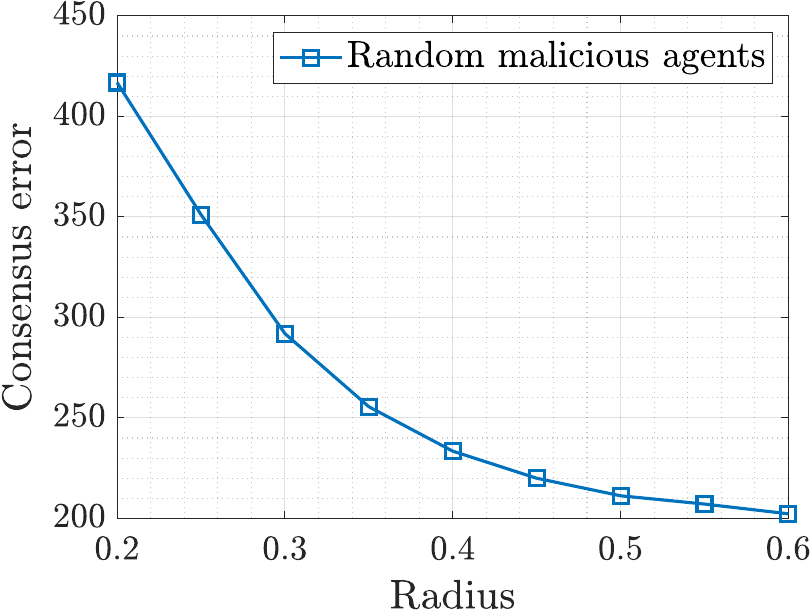}
	\end{minipage}%
	\begin{minipage}[r]{.5\linewidth}
		\centering
		\includegraphics[width=\linewidth]{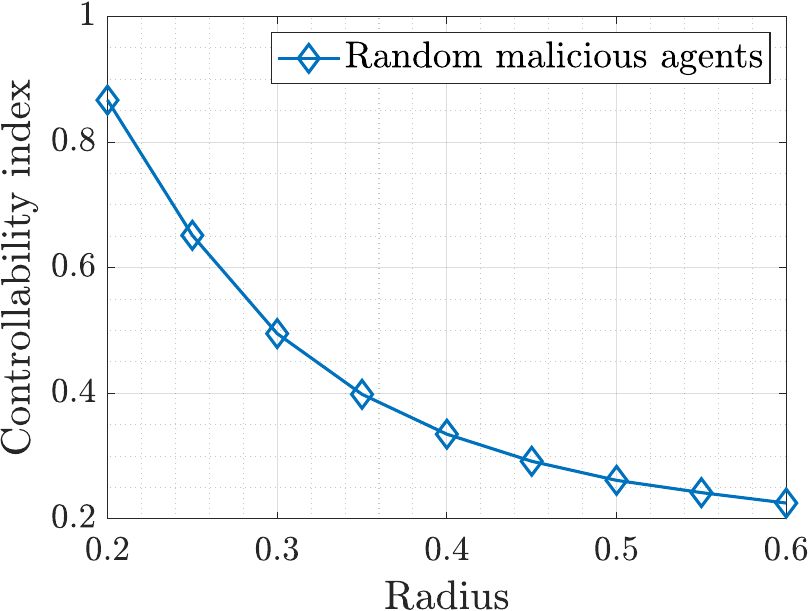}
	\end{minipage}
	\caption{\review{Average performance metrics for 
		random geometric graphs.} 
	}
	\label{fig:graph-geometric}
\end{figure}

%

The main insight is that,
on average,
increasing the graph connectivity mitigates attacks
with respect to both metrics.
Intuitively, this is because 
high degrees mean many interactions among regular \nodes
that the misbehaving \node cannot control directly.
\review{The only remarkable difference is noted in random geometric graphs
with the worst-case misbehaving \node
(top-left box in~\autoref{fig:graph-geometric}),
for which increasing the radius from $0.35$ to $0.5$ also increases the consensus error.
This might be due to the formation of hubs,
that is,
densely connected areas that emerge and become denser as the radius increases,
which an adversary can exploit to quickly spread damage to a large portion of the network.
Notably,
this phenomenon is absent both for the same class of graphs with
random selection of misbehaving \nodes (bottom-left box of~\autoref{fig:graph-geometric})
and in the case of Erd\"os-R\'enyi random graphs (\autoref{fig:graph-erdos_renyi}),
which also typically feature some dense areas
-- even though not with the small world structure typical of random geometric graphs,
see~\cref{fig:erdosrenyi_p4overN__graph,fig:geom_rad025_mal20_graph}.
A deeper study of this phenomenon is an interesting direction of future research.}

\review{Besides density and number of links,
an aspect that also seems to play a role in resiliency is degree balance among nodes.
This can be somehow deduced by the plots referred to the same selection strategy of misbehaving \node:
for example,
with worst-case misbehaving \nodes,
regular graphs exhibit the smallest costs,
random geometric graphs
-- where usually nodes have similar number of neighbors --
yield worse performance,
and Erd\"os-R\'enyi random graphs
-- where both highly connected and almost isolated nodes coexist --
have the largest costs.}

\revision{\review{To more carefully investigate how performance varies with degree balance},
we consider \emph{almost-regular} graphs,
namely, where nodes have degree either $ \Delta $ or $ \Delta - 1 $ for some $ \Delta $.
This corresponds to \textquotedblleft middle-ways" between $ \Delta $- and $ (\Delta-1) $-regular graphs,
which could be ideally placed between two consecutive ticks (degrees) $ \Delta $ and $ \Delta - 1 $ on the $ x $-axis of~\autoref{fig:graph-regular}.

\review{More specifically, 
starting from a $ \Delta $-regular graph, 
we iteratively remove one edge at a time so as to minimize performance degradation
while selecting the worst-case misbehaving \node at each time.
This amounts to removing the edge $e$ that solves}
\begin{align}
	&\min_{e\in\mathcal{E}}\max_{m\in\sensSet} \;\; \ereg(\mathcal{E}\setminus\{e\}), \label{eq:optimization-network-err-remove-one-edge}\\
	&\min_{e\in\mathcal{E}}\max_{m\in\sensSet} \;\; \tr{\Gram{K}(\mathcal{E}\setminus\{e\})}, \label{eq:optimization-network-contr-remove-one-edge}
\end{align}
\review{where $ \mathcal{E} $ is the set of edges (nonzero elements of $ W $)
and we set $ W $ with uniform weights after each removal.}
To get almost-regular graphs,
we remove at most one edge per node.}

\begin{figure}
	\centering
	\begin{minipage}[l]{.5\linewidth}
		\centering
		\includegraphics[width=\linewidth]{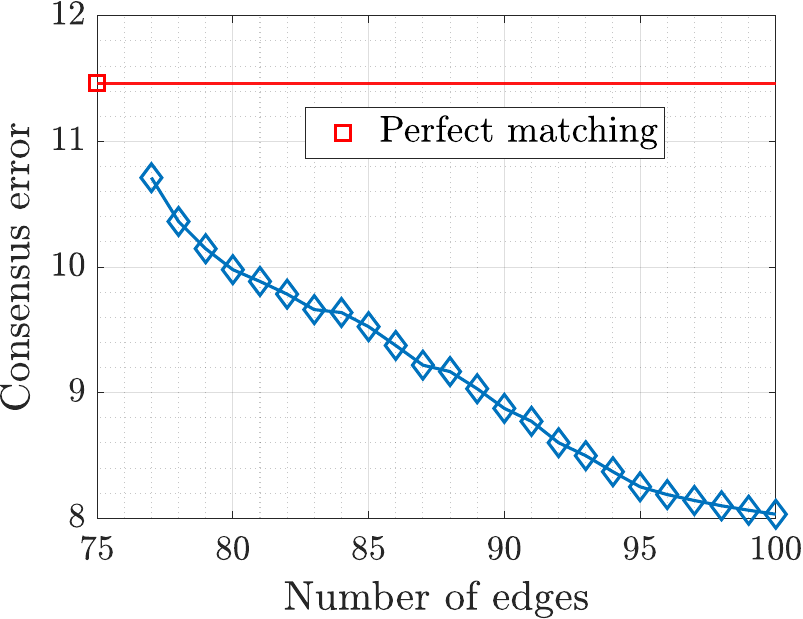}
	\end{minipage}%
	\begin{minipage}[r]{.5\linewidth}
		\centering
		\includegraphics[width=\linewidth]{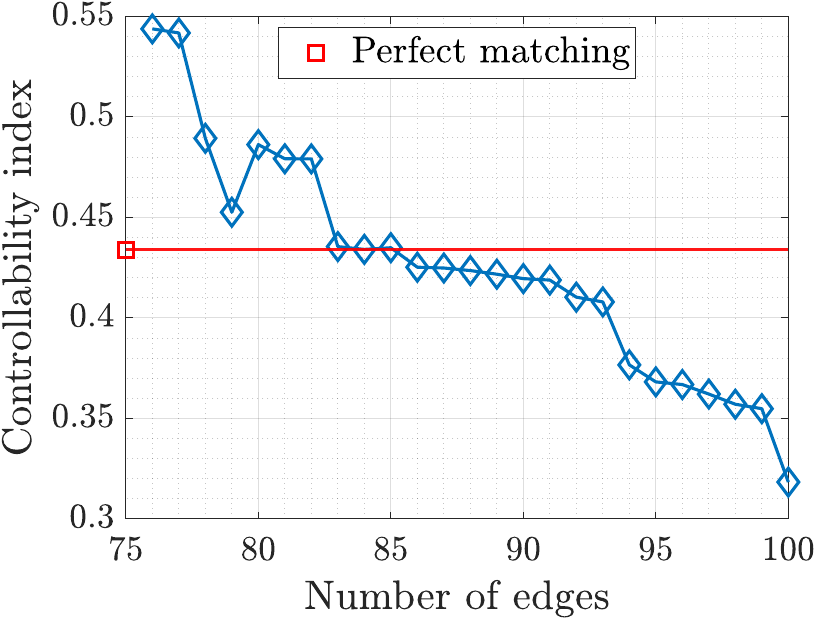}
	\end{minipage}
	\caption{Consensus error (left) and controllability index (right) for almost-regular graphs
		starting from a $ 4 $-regular graph with $ \lam = 0.2 $.
		Edge removal proceeds from right (initially, all $ 100 $ edges are present)
		towards left.
		At each iteration,
		one edge is removed so as to minimize performance degradation
		according to~\eqref{eq:optimization-network-err-remove-one-edge}--\eqref{eq:optimization-network-contr-remove-one-edge}
		while enforcing that each node has degree either three or four.
		At the last iteration (leftmost diamonds),
		most or all nodes have degree three,
		with possibly a few nodes left with degree four.
		The red squares show the performance metrics for a $ 3 $-regular graph
		obtained by removing a perfect matching (set of edges) from the initial $ 4 $-regular graph.
	}
	\label{fig:net-opt-non-regular-lam-0.2}
\end{figure}
\begin{figure}
	\centering
	\begin{minipage}[l]{.5\linewidth}
		\centering
		\includegraphics[width=\linewidth]{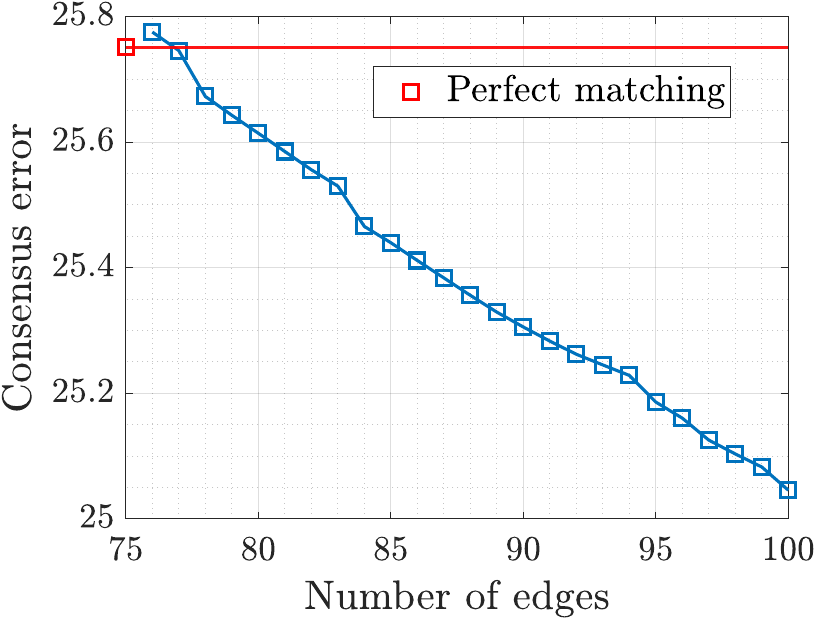}
	\end{minipage}%
	\begin{minipage}[r]{.5\linewidth}
		\centering
		\includegraphics[width=\linewidth]{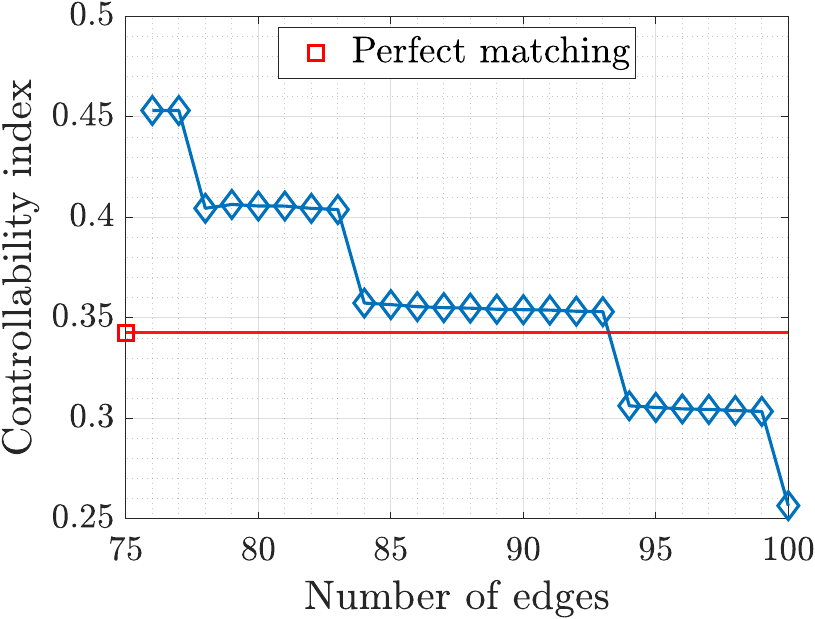}
	\end{minipage}
	\caption{Consensus error (left) and controllability index (right) for almost-regular graphs
		starting from a $ 4 $-regular graph with $ \lam = 0.7 $.}
	\label{fig:net-opt-non-regular-lam-0.7}
\end{figure}

Figures~\ref{fig:net-opt-non-regular-lam-0.2}--\ref{fig:net-opt-non-regular-lam-0.7}
show the performance obtained starting from a $ 4 $-regular graph with $ 50 $ nodes
($ 100 $ edges in total, 
corresponding to the rightmost point in the plots) and gradually pruning edges according to~\eqref{eq:optimization-network-err-remove-one-edge}--\ref{eq:optimization-network-contr-remove-one-edge}
(proceeding leftwards on the $ x $-axis).
Also, performance with a $ 3 $-regular graphs
obtained by removing perfect matchings from the initial $ 4 $-regular graphs
are shown for comparison.\footnote{
	A matching is a set of edges that do not share nodes.
	A maximum matching is a matching of maximal cardinality, and a perfect matching is a maximum matching such that each node is incident to one edge (\emph{total coverage}). 
	Note that our edge removal strategy need not remove exactly one edge for each node in the graph,
	because we constrain the resulting graphs to be almost regular.
	For example,
	the iterative removal corresponding to~\autoref{fig:net-opt-non-regular-lam-0.7}
	stops before reaching a $3$-regular graph.
}
Remarkably, performance degrades (almost) monotonically for both performance metrics as edges are removed. 
This may be explained by a combination of lower connectivity
and degree unbalance,
which 
allows the adversary to exploit highly connected \nodes to make more effective damage against low-connected regular \nodes.

\review{Interestingly, while the consensus error increases smoothly as edges are removed,
the controllability index exhibits \textquotedblleft jumps".
This is evident with large $ \lam $,
as~\autoref{fig:net-opt-non-regular-lam-0.7} shows.
Such a behavior suggests the presence of critical subsets of edges
and might give indication about critical links to be kept or removed.}

Further, in almost all tests (not shown here in the interest of space),
the $ 3 $-regular graph obtained by removing a perfect matching
yielded better performance compared to the last edge removal
(leftmost marker on the blue curve).
\review{This suggests that increasing connectivity may not be beneficial
if it entails less degree balance:
in~\autoref{fig:net-opt-non-regular-lam-0.7},
the $ 3 $-regular graph 
reduces both the consensus error and the controllability index w.r.t. the last graphs obtained by pruning edges
(leftmost markers),
which have one node with degree $ 4 $ and all others with degree $ 3 $.
In particular,
the latter metric is reduced by $ 22\% $ and is comparable to graphs having most nodes with degree $ 4 $.}
However, as shown in~\autoref{fig:net-opt-non-regular-lam-0.2},
a regular graph of degree $ \Delta-1 $ obtained by
removing a perfect matching (not related to performance metrics)
from a $ \Delta $-regular graph
may yield worse performance than almost-regular graphs.
This gives further insight: \review{an arbitrary edge selection
may perform substantially worse compared to a task-related strategy}.

\section{Comparison with Existing Literature}\label{sec:literature-comparison}

In this section, we test our proposed protocol 
and compare its performance with other approaches in the literature.

Many techniques have been proposed 
to mitigate misbehaving \nodes.
However,
they usually focus on reaching a generic consensus,
possibly while keeping the states of regular \nodes within a safe region (usually defined by initial conditions),
and do not consider \textit{\review{performance of}} \emph{average consensus},
which here is key to the distributed optimization task,
as argued in~\autoref{sec:setup}. 
Indeed,
most resilient consensus strategies aim to make the regular \nodes agree on, 
\eg a common location (such as in robot gathering) in the face of misleading interactions, 
but need not relate the consensus value to the initial locations.

We compare two strategies: Weighted Mean Subsequence Reduced (W-MSR)~\cite{6481629}
and Secure Accepting and Broadcasting Algorithm (SABA)~\cite{8814959}.
As noted in~\autoref{sec:intro-related-literature},
many resilient algorithms 
adapt W-MSR to specific applications
and enjoy the same guarantees. 
W-MSR suffers from two main limitations 
related to $ r $-robustness, 
which is the cornerstone of all theoretical analysis.
First, while sufficient conditions for resilient consensus are clear,
there is little clue about necessary conditions.
This translates into an unknown behavior of the system
if $r$-robustness does not hold.
While $ r $-robustness has proved a good characterization for update rules based on W-MSR,
it raises practical limitations.
On the one hand, the communication network 
may be fixed but not robust enough.
On the other hand,
checking $ r $-robustness 
is computationally intractable for large-scale networks~\cite{7447011}.
Thus, in some cases, for example with a sparse structure,
a more conservative behavior with provable performance bounds may be preferred.
Also, W-MSR requires to estimate the number of misbehaving \nodes
affecting the network. 
This may be an issue:
if the estimate is too low, 
regular \nodes may be deceived and average consensus disrupted,
whereas,
if it is too high,
the updates may be too conservative,
possibly preventing convergence. 
Further, misbehaviors could happen in a time-varying fashion and make the $r$-robustness fail at times,
yielding poor performance overall.
SABA does not estimate the number of misbehaving \nodes,
but stores all received values in a buffer
and processes them with a voting strategy.
However, 
this design may impose impractical memory requirements, 
and the convergence of SABA is still ensured under a minimal $ r $-robustness.

In the next simulations,
we consider $N=100$ \nodes interacting through sparse communication networks,
whose low connectivity hampers W-MSR and SABA,
and matrices $ W^o $ with homogeneous weights.
As performance metric,
we computed the objective cost of the distributed optimization task~\eqref{eq:network-cost-regular}, 
which equals $ \ereg $ up to additive constants,
cf.~\autoref{sec:performance-metric}.
The observations are drawn as $ \priornode{}\sim\gauss(0,0.1I)$
and each misbehaving \node $ m $ {is assigned a deception bias $ v_m\in[2,6] $}.
For each scenario,
we chose the parameter $ \lam $
by selecting the minimizer of the theoretical error $ \ereg(\lambda) $ with ${V = 5I_M}$.

\extended{}{\begin{figure}
		\centering
		\begin{subfigure}{0.5\linewidth}
			\centering
			\includegraphics[height=.78\linewidth]{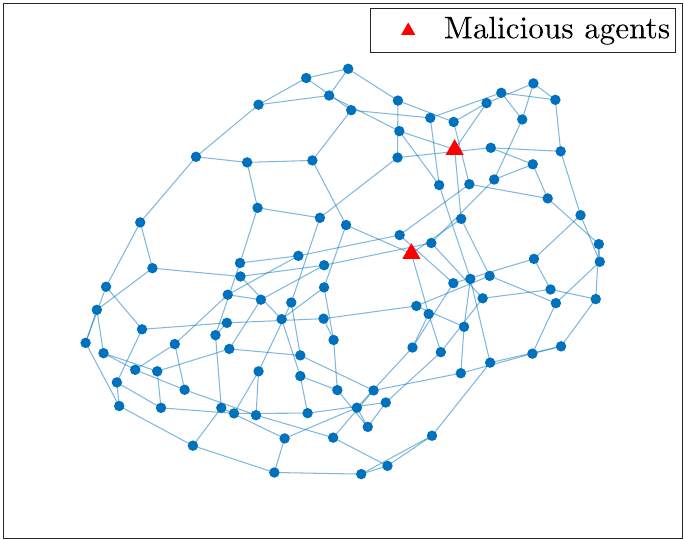}
			\caption{Communication network.}
			\label{fig:reg(3)_graph}
		\end{subfigure}%
		\hfill
		\begin{subfigure}{0.5\linewidth}
			\centering
			\includegraphics[height=.78\linewidth]{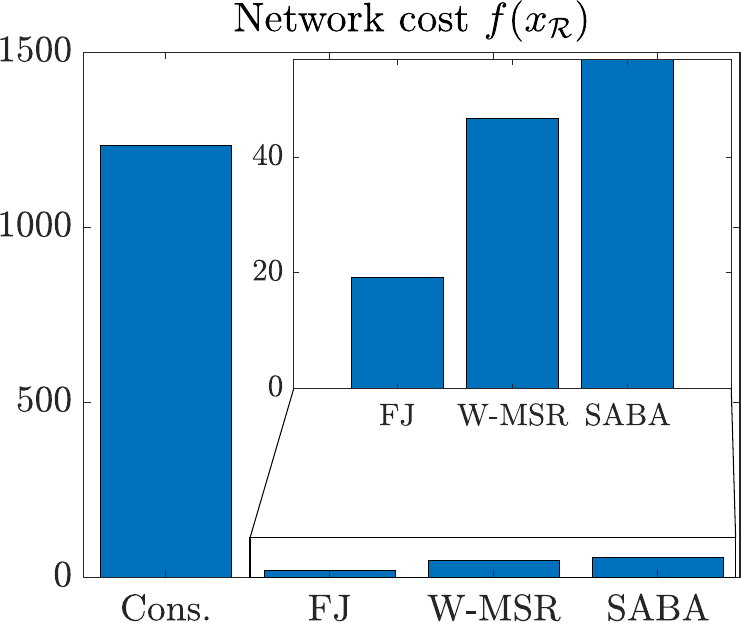}
			\caption{Cost~\eqref{eq:network-cost-regular} of regular \nodes.}
			\label{fig:reg(3)_net_cost}
		\end{subfigure}
		\caption{Comparison among consensus,
			FJ,
			W-MSR~\cite{6481629},
			and SABA~\cite{8814959} with $ 3 $-regular graph and two misbehaving \nodes.
		}
		\label{fig:reg(3)}
	\end{figure}
	
	Figure~\ref{fig:reg(3)} illustrates a network 
	where \nodes interact on a $ 3 $-regular graph (\autoref{fig:reg(3)_graph})
	with two misbehaving \nodes (red triangles). 
	Importantly,
	$ 3 $-regular graphs are not $ r $-robust enough to tolerate misbehaving \nodes,
	and therefore theoretical guarantees of MSR-based approaches do not hold. 
	We implement W-MSR assuming that each regular \node has at most one misbehaving neighbor,
	because larger values make updates trivial,
	\ie $ \xnode{i}{k} \equiv \xnode{i}{0} $. 
	Such limitations allow dynamics~\eqref{eq:FJ-dynamics} to outperform both W-MSR and SABA,
	as shown in~\autoref{fig:reg(3)_net_cost}.}

\begin{figure}
	\centering
	\begin{subfigure}{0.5\linewidth}
		\centering
		\includegraphics[height=.78\linewidth]{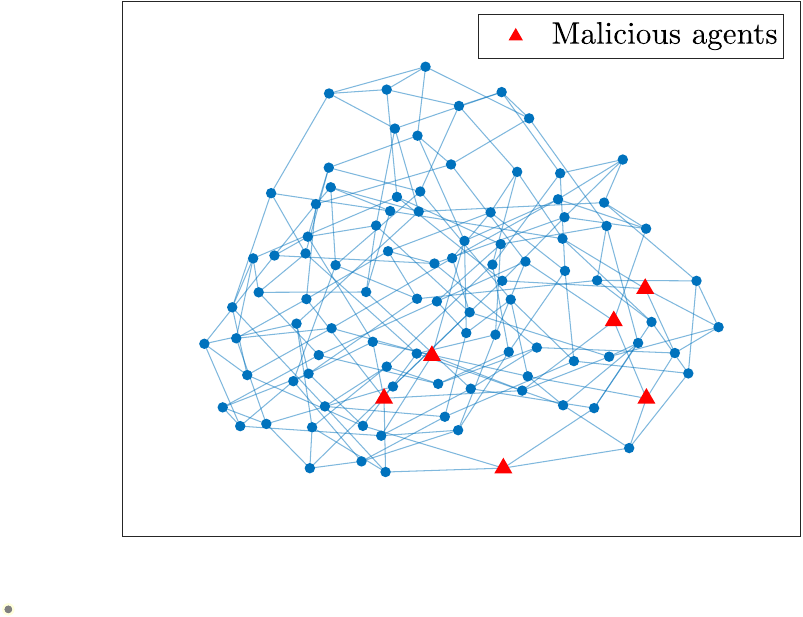}
		\caption{Communication network.}
		\label{fig:reg(4)_graph}
	\end{subfigure}%
	\hfill
	\begin{subfigure}{0.5\linewidth}
		\centering
		\includegraphics[height=.78\linewidth]{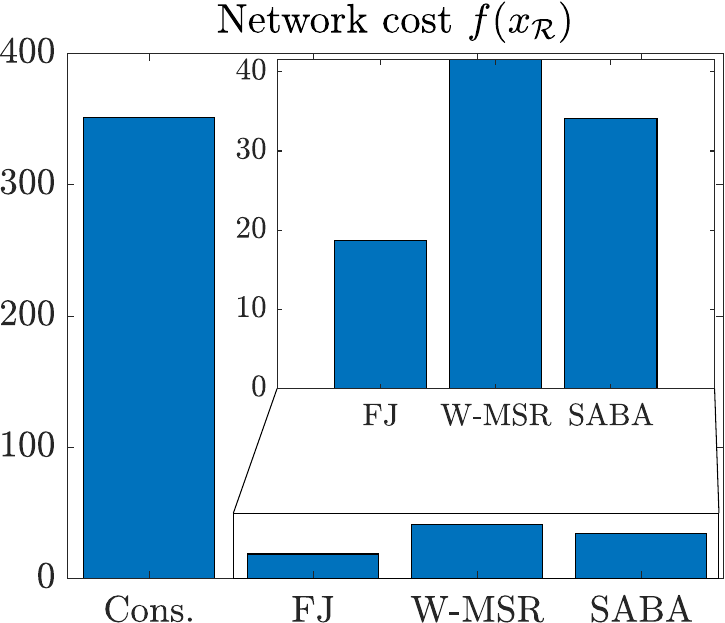}
		\caption{Cost~\eqref{eq:network-cost-regular} of regular \nodes.}
		\label{fig:reg(4)_net_cost}
	\end{subfigure}
	\caption{Comparison among consensus,
		FJ,
		W-MSR~\cite{6481629},
		and SABA~\cite{8814959} with $ 4 $-regular graph and six misbehaving \nodes.
	}
	\label{fig:reg(4)}
\end{figure}

\extended{
	Figure~\ref{fig:reg(4)} illustrates a network 
	where the agents use a
}{In our second experiment,
	we use a denser,}
regular graph with degree $\Delta=4$ as communication network 
with six misbehaving \nodes (\autoref{fig:reg(4)_graph}).
\extended{We implement W-MSR assuming that each regular \node has (at most) one misbehaving neighbor
	because larger values make updates trivial,
	\ie $ \xnode{i}{k} \equiv \xnode{i}{0} $.}{}
However,
some misbehaving \nodes communicate with the same regular \nodes
(\eg the two in the bottom-right portion of the graph),
making this scenario challenging for W-MSR and SABA
whose $r$-robustness requirement suffers the sparse communication graph.
While both SABA and W-MSR perform poorly (\autoref{fig:reg(4)_net_cost}),
our approach mitigates the attacks by setting $ \lam $ at a suitably large value.

\extended{}{\begin{figure}
		\centering
		\begin{subfigure}{0.5\linewidth}
			\centering
			\includegraphics[height=.78\linewidth]{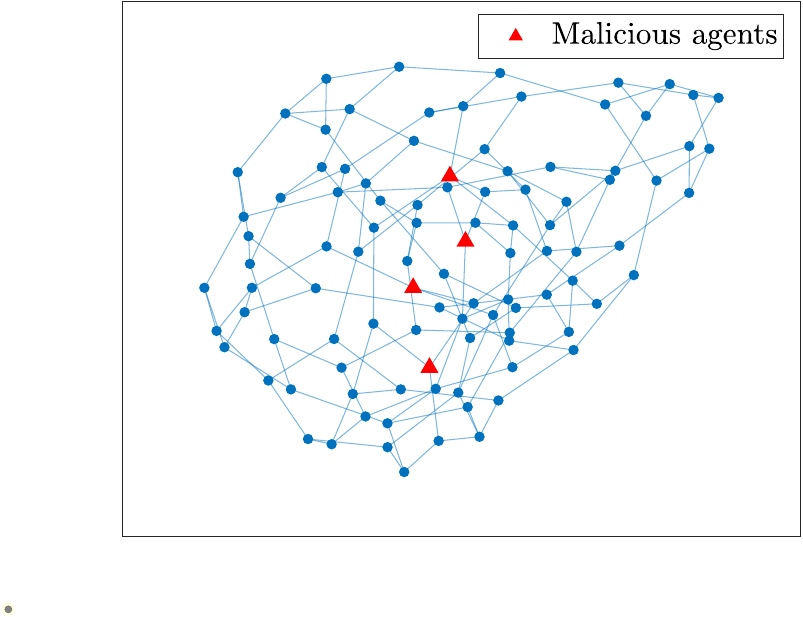}
			\caption{Communication network.}
			\label{fig:reg(3,4)_graph}
		\end{subfigure}%
		\hfill
		\begin{subfigure}{0.5\linewidth}
			\centering
			\includegraphics[height=.78\linewidth]{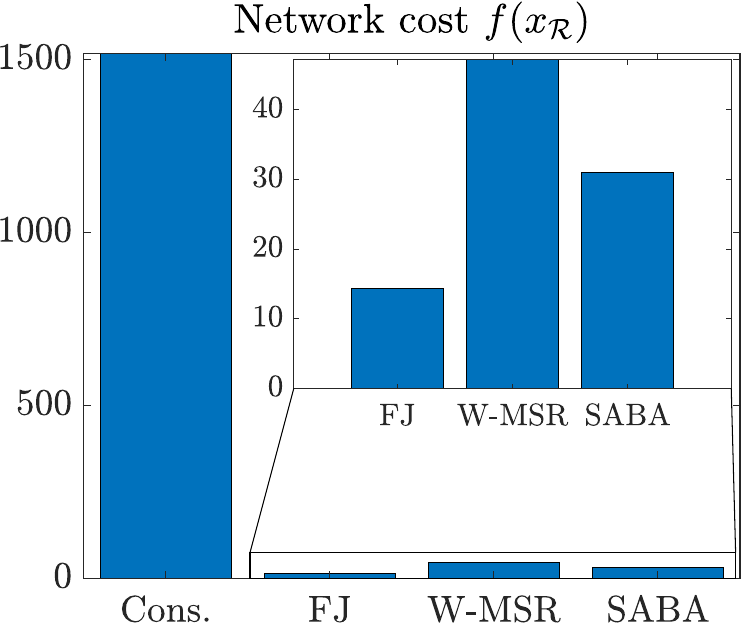}
			\caption{Cost~\eqref{eq:network-cost-regular} of regular \nodes.}
			\label{fig:reg(3,4)_net_cost}
		\end{subfigure}
		\caption{Comparison among consensus,
			FJ,
			W-MSR~\cite{6481629},
			and SABA~\cite{8814959} with \review{almost-regular} graph \review{with degrees $(3,4)$} and four misbehaving \nodes.
		}
		\label{fig:reg(3,4)}
	\end{figure}
	In~\autoref{fig:reg(3,4)},
	we consider a network where nodes have degree three or four (\autoref{fig:reg(3,4)_graph})
	and $ W $ is row stochastic. 
	In this case, one may question whether a doubly-stochastic matrix
	could improve performance of the standard consensus protocol,
	in light of its optimality under nominal conditions.
	However, in the presence of misbehaving \nodes,
	standard consensus always converges to the average of the misbehaving states
	regardless of weights in $ W $ (cf.~\cref{ass:mal-node-dynamics} and~\eqref{eq:cons-err-mal} in Appendix).
	Conversely,~\autoref{fig:reg(3,4)_net_cost} shows that dynamics~\eqref{eq:FJ-dynamics} 
	is robust against misbehaving \nodes
	even though it cannot retrieve the optimal solution under nominal conditions.}

\extended{}{\begin{figure}
		\centering
		\begin{subfigure}{0.5\linewidth}
			\centering
			\includegraphics[height=.78\linewidth]{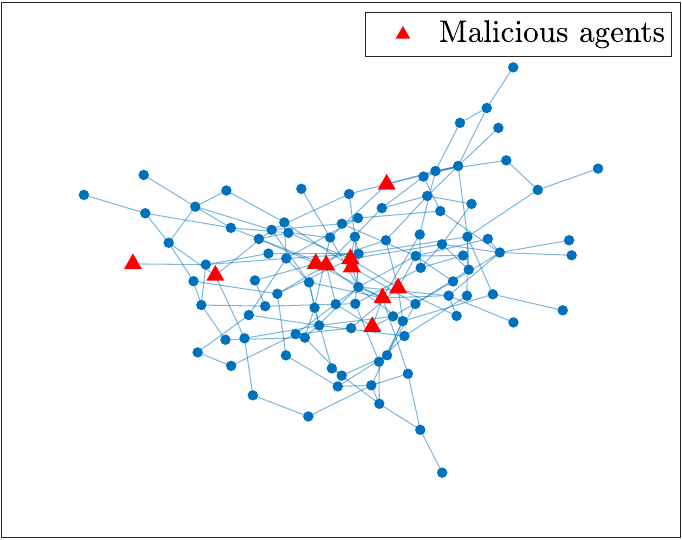}
			\caption{Communication network.}
			\label{fig:erdosrenyi_p3overN__graph}
		\end{subfigure}%
		\begin{subfigure}{0.5\linewidth}
			\centering
			\includegraphics[height=.78\linewidth]{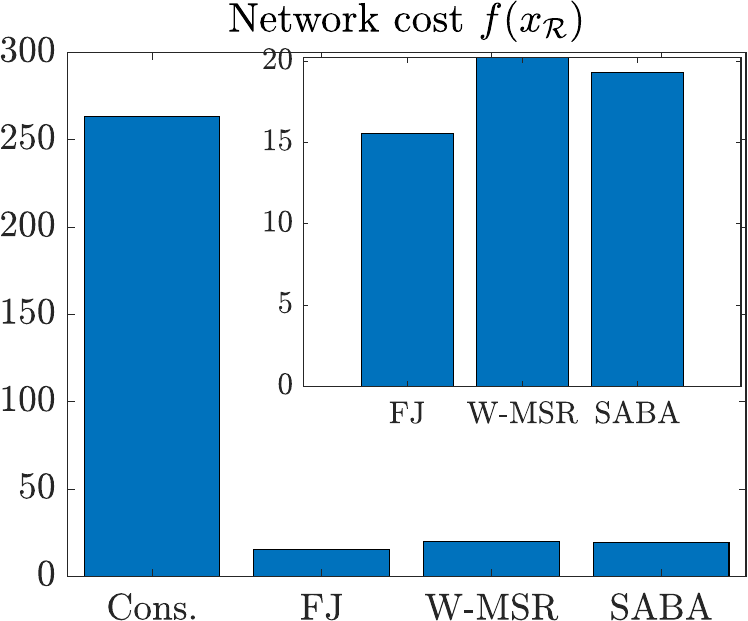}
			\caption{Cost~\eqref{eq:network-cost-regular} of regular \nodes.}
			\label{fig:erdosrenyi_p3overN__net_cost}
		\end{subfigure}
		\caption{\review{Comparison among consensus,
				FJ,
				W-MSR~\cite{6481629},
				and SABA~\cite{8814959} with Erd\"os-R\'enyi random graph with $p=\nicefrac{3}{N}$ and ten misbehaving \nodes.}
		}
		\label{fig:erdosrenyi_p3overN}
\end{figure}}

\begin{figure}
	\centering
	\begin{subfigure}{0.5\linewidth}
		\centering
		\includegraphics[height=.78\linewidth]{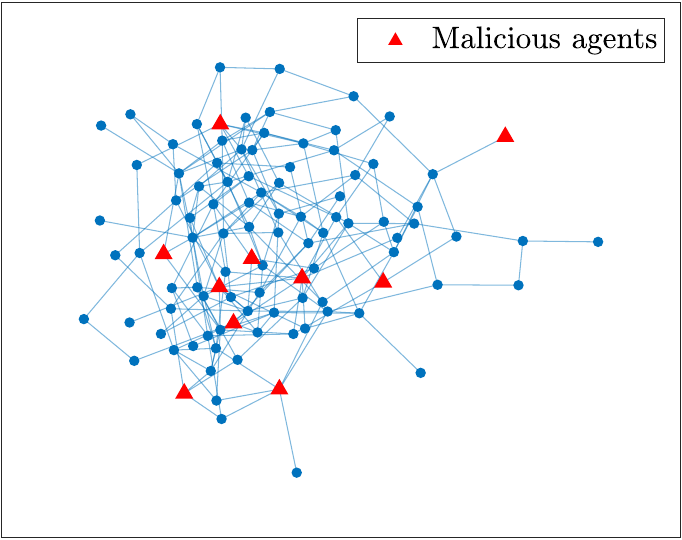}
		\caption{Communication network.}
		\label{fig:erdosrenyi_p4overN__graph}
	\end{subfigure}%
	\begin{subfigure}{0.5\linewidth}
		\centering
		\includegraphics[height=.78\linewidth]{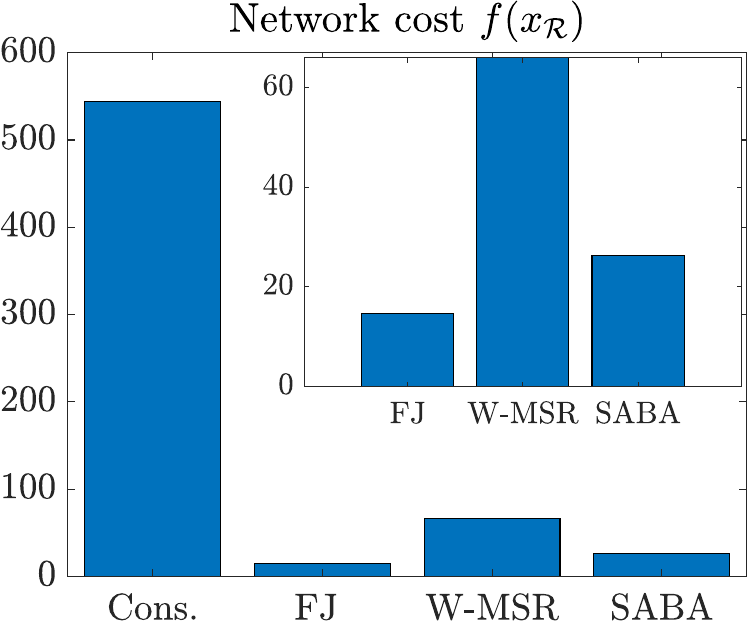}
		\caption{Cost~\eqref{eq:network-cost-regular} of regular \nodes.}
		\label{fig:erdosrenyi_p4overN__net_cost}
	\end{subfigure}
	\caption{\review{Comparison among consensus,
			FJ,
			W-MSR~\cite{6481629},
			and SABA~\cite{8814959} with Erd\"os-R\'enyi random graph with $p=\nicefrac{4}{N}$ and ten misbehaving \nodes.}
	}
	\label{fig:erdosrenyi_p4overN}
\end{figure}

\review{
	\extended{In~\cref{fig:erdosrenyi_p4overN},}{In~\cref{fig:erdosrenyi_p3overN,fig:erdosrenyi_p4overN},}
	we simulate the protocols over \extended{an}{two} Erd\"os-R\'enyi random \extended{graph}{graphs}
	with link probability \extended{}{$p=\nicefrac{3}{N}$ and} $p=\nicefrac{4}{N}$\extended{.}{,
		respectively}
	(hence,
	each agent has $(N-1)p$ neighbors on average),
	and ten misbehaving \nodes ($10\%$ of the total number of \nodes).
	Note that the matrix $W$ is row stochastic.
	\extended{Also in this case,}{In both cases,}
	the dynamics~\eqref{eq:FJ-dynamics} tames the numerous attacks
	better than the confronted approaches.
}

\extended{}{\begin{figure}
		\centering
		\begin{subfigure}{0.5\linewidth}
			\centering
			\includegraphics[height=.78\linewidth]{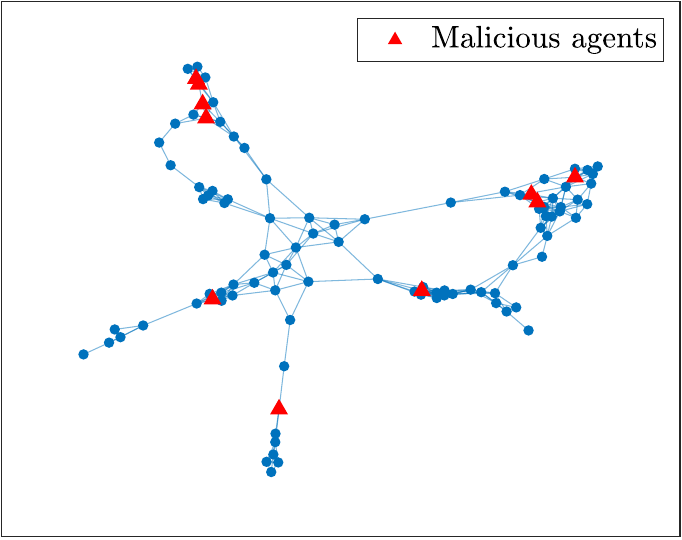}
			\caption{Communication network.}
			\label{fig:geom_rad015_mal10_graph}
		\end{subfigure}%
		\begin{subfigure}{0.5\linewidth}
			\centering
			\includegraphics[height=.78\linewidth]{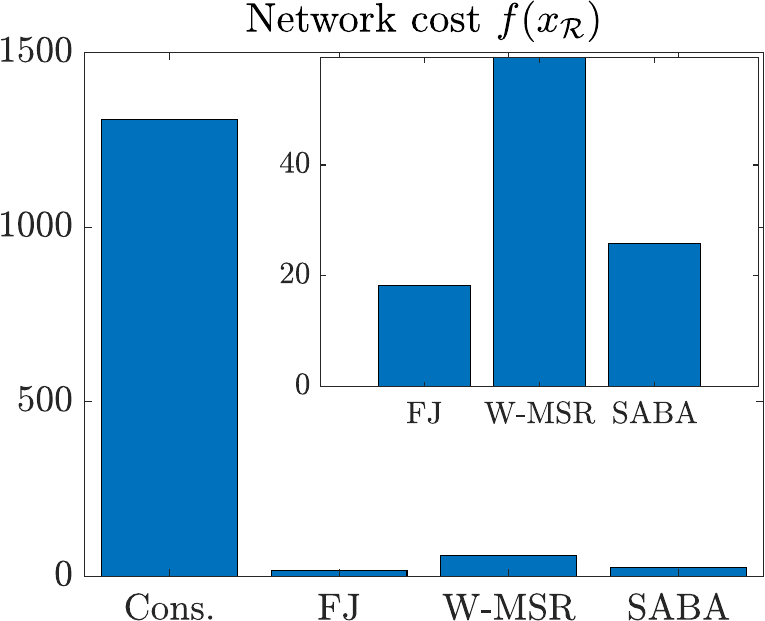}
			\caption{Cost~\eqref{eq:network-cost-regular} of regular \nodes.}
			\label{fig:geom_rad015_mal10_net_cost}
		\end{subfigure}
		\caption{\review{Comparison among consensus,
				FJ,
				W-MSR~\cite{6481629},
				and SABA~\cite{8814959} with random geometric graph with $\rho=0.15$ and ten misbehaving \nodes.}
		}
		\label{fig:geom_rad015_mal10}
	\end{figure}
	
	\begin{figure}
		\centering
		\begin{subfigure}{0.5\linewidth}
			\centering
			\includegraphics[height=.78\linewidth]{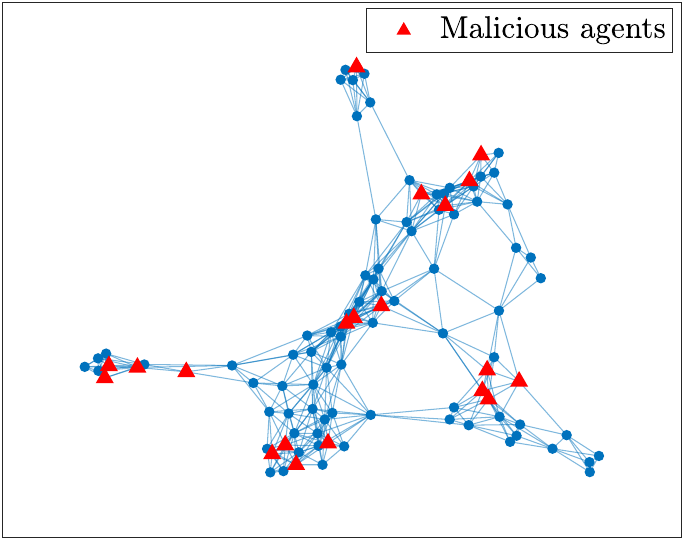}
			\caption{Communication network.}
			\label{fig:geom_rad02_mal20_graph}
		\end{subfigure}%
		\begin{subfigure}{0.5\linewidth}
			\centering
			\includegraphics[height=.78\linewidth]{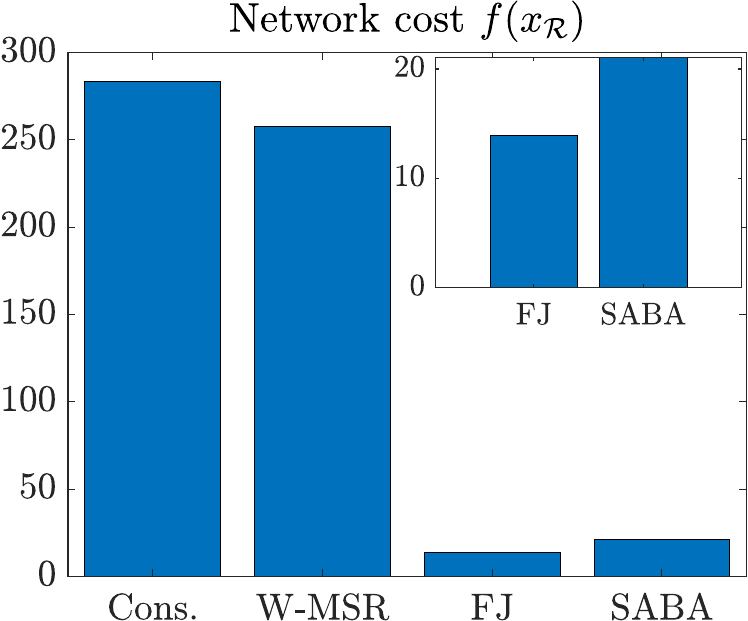}
			\caption{Cost~\eqref{eq:network-cost-regular} of regular \nodes.}
			\label{fig:geom_rad02_mal20_net_cost}
		\end{subfigure}
		\caption{\review{Comparison among consensus,
				FJ,
				W-MSR~\cite{6481629},
				and SABA~\cite{8814959} with random geometric graph with $\rho=0.2$ and twenty misbehaving \nodes.}
		}
		\label{fig:geom_rad02_mal20}
\end{figure}}

\begin{figure}
	\centering
	\begin{subfigure}{0.5\linewidth}
		\centering
		\includegraphics[height=.78\linewidth]{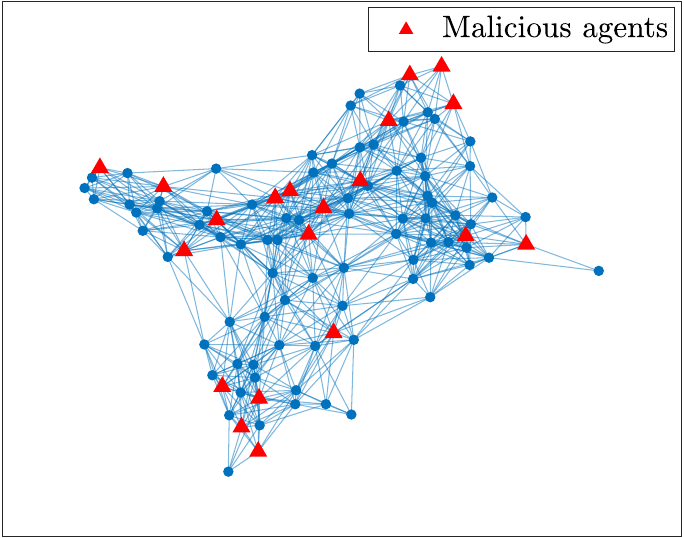}
		\caption{Communication network.}
		\label{fig:geom_rad025_mal20_graph}
	\end{subfigure}%
	\begin{subfigure}{0.5\linewidth}
		\centering
		\includegraphics[height=.78\linewidth]{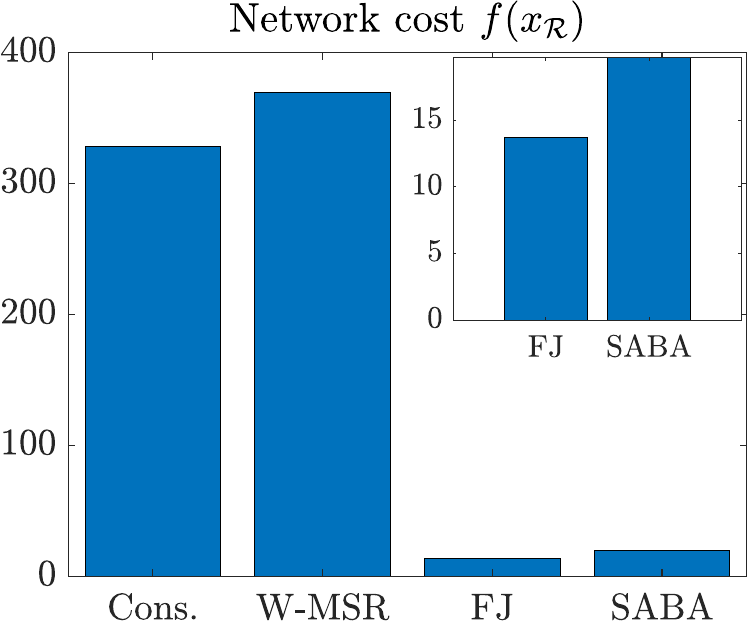}
		\caption{Cost~\eqref{eq:network-cost-regular} of regular \nodes.}
		\label{fig:geom_rad025_mal20_net_cost}
	\end{subfigure}
	\caption{\review{Comparison among consensus,
			FJ,
			W-MSR~\cite{6481629},
			and SABA~\cite{8814959} with random geometric graph with $\rho=0.25$ and twenty misbehaving \nodes.}
	}
	\label{fig:geom_rad025_mal20}
\end{figure}

\review{
	Finally,
	we address \extended{a}{} random geometric \extended{graph}{graphs}
	\extended{with radius $\rho=0.25$
		and twenty misbehaving \nodes in~\autoref{fig:geom_rad025_mal20}.
	}{with several radii ($0.15$ in~\autoref{fig:geom_rad015_mal10},
		$0.20$ in~\autoref{fig:geom_rad02_mal20},
		and $0.25$ in~\autoref{fig:geom_rad025_mal20})
		and increasing amounts of misbehaving \nodes
		to overcome the higher density of the network.}
	Also in this case,
	the matrix $W$ is row stochastic.
	Interestingly,
	W-MSR is rather challenged by this class of graphs,
	yielding \extended{a}{} large \extended{cost}{costs}.
	On the other hand,
	the dynamics~\eqref{eq:FJ-dynamics} again manages to keep the error small
	compared to the other algorithms.
}

\extended{\review{Many other simulations that validated our approach are not reported here in the interest of space
		and can be found in~\cite{arxiv}.}
}{}

\begin{rem}[Advantages of FJ dynamics]
	The experiments above highlight some advantages of the proposed approach.
	Firstly, the tunable parameter $\lam$
	makes the algorithm flexible,
	because it can smoothly adapt to a different attack intensity 
	while still providing decent performance.
	Further, while the optimal parameterization
	requires exact knowledge of the adversary,
	which may not be reasonably assumed,
	yet our proposed approach proves pretty robust to the choice of a specific $ \lam $,
	\review{as shown in Figs.~\ref{fig:err_FJ_reg(3)_d10-100_1mal_Varn_10^(-.2)}--\ref{fig:err_FJ_reg(3,4)_d10_mal1-4_Varn_10^(-.2)}
		where the error is kept small around $\lam^*$.}
	This also holds with row-stochastic matrices,
	enabling simple weighing rules to be locally implemented.
	In contrast, in other approaches the cost function may be highly sensitive to
	some design parameters,
	\eg the estimated number of misbehaving \nodes in W-MSR.
	Further, most works in the literature do not describe the
	system behavior when resilient consensus is not guaranteed.
	In fact, they usually either ensure that the states of the agents remain
	inside the convex hull of the initial conditions
	(which may be equivalent to setting $ \lam  = 1 $ in~\eqref{eq:FJ-dynamics}), or let agents
	reach consensus but potentially be steered far away from
	initial conditions~\cite{8798516}.
	Finally, computational complexity and memory requirements are minimal,
	which is typically desired for resource-constrained devices.
\end{rem}

\section{\titlecap{conclusion and future work}}\label{sec:conclusion}

In this article,
we have proposed a competition-based protocol based on the Friedkin-Johnsen dynamics
to mitigate a class of misbehavior that disrupts a quadratic distributed optimization task.
We have presented formal results and numerical experiments on performance and optimal parametrization,
and showed that our approach can outperform state-of-the-art algorithms. 
Further, 
we have discussed the competition-collaboration trade-off with analytical arguments
that are insightful towards a deeper understanding of the fundamental properties of the system
in the presence of misbehaviors.
Finally, we have addressed network design and
\review{explored how resilience relates to graph connectivity},
looking both at the optimization performance and at the energy spend to misbehave.

This approach opens several avenues for future research.
Firstly,
it is desirable to address an effective design
of parameters $ \lam_i $'s in the realistic case where knowledge about the attack is scarce.
This may also involve online reweighing of protocol parameters,
for example in the realm of recent work
where weights are updated
via trust information or evidence theory~\cite{8798516,9468419}.

Secondly, the more general and challenging scenario of distributed optimization
should be addressed.
In this case, a common approach is to alternate local descent steps 
to consensus updates to steer all \nodes towards a common point\revision{~\cite{9241497,7134753}}.
Here, the task-tailored descent steps may critically impact performance
even if consensus steps are made resilient.

A third research avenue involves zero-sum games
to model interactions among agents~\cite{8788623,8567999}.
In particular, in asymmetric zero-sum games,
one player has more knowledge than the other,
which is a suitable model for worst-case attacks.
In this case,
a relevant challenge is determining the optimal strategies for both players,
to ultimately derive effective resilient algorithms in the presence of intelligent adversaries.

Finally,
it is interesting to deeply investigate the design of the communication network.
While graph robustness to node or edge failures has been extensively addressed~\cite{1431045,SOHOUENOU2020100353,8885721,NEURIPS2019_e2f374c3},
the novel element given by the dynamics~\eqref{eq:FJ-dynamics} calls for a tailored investigation
as heuristically motivated in~\autoref{sec:network-optimization}.
Also, in the spirit of a graph-theoretic approach,
a comparison between classical centrality measures and worst-case attacks
may be useful to get insight about \nodes that deserve higher attention.
	

	\numberwithin{lemma}{subsection}
	\numberwithin{cor}{subsection}
	\numberwithin{equation}{subsection}
	
	\appendix

\subsection{\titlecap{Useful lemmas}}\label{app:lemmas}

In this Appendix,
we report some standard facts in linear algebra that will be used in the following proofs.
\begin{lemma}\label{lem:derivative-trace}
	Let $\alpha\in\Real{}$ and $ A,B\in\Real{n\times n} $ differentiable functions of $\alpha$,
	then the derivative of $ \tr{A^\top B} $ is
	\begin{equation}\label{eq:derivative-trace}
		\dfrac{d\tr{A^\top B}}{d\alpha} = \tr{\dfrac{dA^\top}{d\alpha}B} + \tr{A^\top\dfrac{dB}{d\alpha}}.
	\end{equation}
\end{lemma}
\begin{lemma}\label{lem:derivative-inverse}
	Let $\alpha\in\Real{}$ and $ A\in\Real{n\times n} $ invertible and differentiable function of $\alpha$,
	then the derivative of $ A\inv $ is 
	\begin{equation}\label{eq:derivative-inverse}
		\dfrac{dA\inv}{d\alpha} = -A\inv\dfrac{dA}{d\alpha}A\inv.
	\end{equation}
\end{lemma}
\begin{lemma}\label{lem:eigendecomposition-inverse}
	Let $ A\in\Real{n\times n} $ invertible with eigenpair $ (\lambda,v) $,
	then $ A^{-1} $ has eigenpair  $ (\lambda^{-1},v) $.
\end{lemma}
\begin{cor}\label{cor:A-inv-diagonalizable}
	If $ A\in\Real{n\times n} $ is diagonalizable,
	then $ A $ and $ A\inv $ are simultaneously diagonalizable.
\end{cor}
\begin{lemma}\label{lem:eigendecomposition-(I-A)}
	Let $ A\in\Real{n\times n} $ have eigenpair $ (\lambda,v) $,
	then $ (I-\alpha A) $ has eigenpair $ \left((1-\alpha\lambda),v\right) $.
\end{lemma}
	\review{

\subsection{Proof of~\cref{lem:P-monotonic}}\label{app:proof-P}

We use the implicit function theorem to prove that each diagonal element $P_{ii}$ of $P$ is strictly decreasing with $\lam$.
Let 
\begin{equation}\label{eq:gii}
	g_i(\lam,P_{ii}) \doteq P_{ii} - (1-\lam)^2\left[\Wreg P\Wreg^\top\right]_{ii} - (1-\lam)^2\tilde{Q}_{ii}
\end{equation}
where $\tilde{Q}=\Wmal Q\Wmal^\top$.
The implicit function theorem holds for the solutions of $g_i(\lam,P_{ii}) = 0$
with $\lam\in(0,1)$:
for $i\in\regSet$,
\begin{equation}\label{eq:g-partial-P}
	\dfrac{\partial g_i(\lam,P_{ii})}{\partial P_{ii}} = 1 - (1-\lam)^2W_{ii}^2 \neq 0 \quad \forall \lam \in (0,1).
\end{equation}
Making dependence on $\lam$ explicit and for $\lam\in(0,1)$,
we get
\begin{equation}\label{eq:P-partial-lam}
	\begin{aligned}
		\dfrac{\d{P_{ii}}(\lam)}{\d{\lam}} 	&= -\dfrac{\partial g_i(\lam,P_{ii}(\lam))}{\partial \lam}\left(\dfrac{\partial g_i(\lam,P_{ii}(\lam))}{\partial P_{ii}(\lam)}\right)^{-1}\\
											&= -\dfrac{2(1-\lam)\left(\left[\Wreg P\Wreg^\top\right]_{ii}+\tilde{Q}_{ii}\right)}{1 - (1-\lam)^2W_{ii}^2} < 0.
	\end{aligned}
\end{equation}
Finally,
from $\eregn(\lam) = \sum_{i\in\regSet}P_{ii}(\lam)$
and linearity of the derivative,
it follows that $\eregn(\lam)$ is decreasing.

For $\lam=1$,
we trivially get $P=0$ and thus $\eregn(1) = 0$.}

\subsection{Proof of~\cref{prop:cons-err-mal}}\label{app:proof-full-comp-vs-full-coll}

\review{In this and all following proofs,
the constant $\kappa$ in~\eqref{eq:cons-error-explicit-1} is neglected for the sake of simplicity.}

We first compute the consensus error with $ \lam = 1 $:
\begin{equation}\label{eq:FJ-err-mal-lamda-1}
	\begin{aligned}
		\review{\ereg(1)} &= \review{\eregv(1)} 
		= \tr{\Varn_{11}(I_R-\Creg)}\\
		&= \dfrac{R-1}{R}\sumreg{i}\varnode{i} - \dfrac{1}{R}\sumreg{i}\sum_{\substack{j\in\regSet\\j\neq i}}\varnode[i]{j}.
	\end{aligned}
\end{equation}

With $\lam=0$,
the average steady-state consensus value is determined by the biased observations of malicious \nodes,
\ie $ \xregbar=\thbarmal \doteq \frac{1}{M}\sum_{m\in\malSet}(\priornode{m} + v_m) $.
We have:
\begingroup
\allowdisplaybreaks
\begin{align*}\label{eq:cons-err-mal}
\review{	\eregv(0)} &= \Enorm{\onereg\thbarmal - \onereg\thbarreg}\\ 
	&= \dfrac{R}{M^2}\summal{m}d_m + \dfrac{R}{M^2}\summal{m}\left(\varnode{m}+\sum_{\substack{n\in\malSet\\n\neq m}}\varnode[m]{n}\right)\\
	&\hspace{5mm} +\dfrac{1}{R}\sumreg{i}\left(\varnode{i}+\sum_{\substack{j\in\regSet\\j\neq i}}\varnode[i]{j}\right) - \dfrac{2}{M}\sumreg{i}\summal{m}\varnode[i]{m}\\
	\review{\eregn(0)} &= \review{\tr{P(0)}}.\stepcounter{equation}\tag{\theequation}
\end{align*}
\endgroup
By comparing the expressions in~\eqref{eq:FJ-err-mal-lamda-1} and~\eqref{eq:cons-err-mal},
it follows that $\ereg(0)>\ereg(1)$ is equivalent to
\begin{multline}\label{eq:full-comp-better-full-comp-condition-proof}
	\review{\dfrac{M^2}{R}\tr{P(0)} +} \sum_{m\in\malSet} d_m > -\summal{m}\left(\varnode{m}+\sum_{\substack{n\in\malSet\\n\neq m}}\varnode[m]{n}\right)\\
	+ \dfrac{M^2}{R}\sumreg{i}\varnode{i}
	-\dfrac{2M^2}{R^2}\sumreg{i}\!\!\left(\!\!\varnode{i}+\sum_{\substack{j\in\regSet\\j\neq i}}\varnode[i]{j}\!\!\right)\! + \dfrac{2M}{R}\sumreg{i}\summal{m}\varnode[i]{m}
	,
\end{multline}
which leads to condition~\eqref{eq:full-comp-better-full-coll-condition}.

\subsection{Proof of~\cref{prop:opt-lam-nontrivial}}\label{app:proof-opt-lam-nontrivial}

From~\eqref{eq:cons-error-explicit-1},
we get
\begin{equation}\label{eq:cons-error-regular-derivative}
	\begin{aligned}
		\dfrac{\d{\ereg}(\lam)}{\d{\lam}} &= \review{\dfrac{\d{\eregv}(\lam)}{\d{\lam}} + \dfrac{\d{\eregn}(\lam)}{\d{\lam}}}\\
		&= \review{\kappa}\dfrac{1}{\lam}\tr{\Var L^\top\left(I-W^\top L^\top\right)\selreg^\top E} \review{+ \dfrac{\d{\eregn}(\lam)}{\d{\lam}}},
	\end{aligned}
\end{equation}
where~\cref{lem:derivative-trace,lem:derivative-inverse} were used \review{and $\kappa>0$}.

\subsubsection{Part one: $ {\lam^* < 1} $}
\review{From~\eqref{eq:P-partial-lam},
	$\frac{\d{\eregn(\lam)}}{\d{\lambda}}\big|_{\lam=1}=0$}
	and
\begin{equation}\label{eq:cons-error-regular-derivative-at-1}
	\dfrac{\d{\ereg}(\lam)}{\d{\lam}}\Big|_{\lam=1} = \kappa\tr{\Var \left(I-W^\top\right)\selreg^\top\left(\selreg - \Creg\selreg\right)}.
\end{equation}
The argument of the trace in~\eqref{eq:cons-error-regular-derivative-at-1} has expression
\begin{equation}\label{eq:cons-error-regular-derivative-at-1-arg}
	\left[\begin{array}{ c | c }
		A & 0 \\ 
		\hline
		\star & 0
	\end{array}\right], \qquad A\in\Real{R\times R}.
\end{equation}
\begin{description}[leftmargin=*]
	\item[\review{Condition C1.}] \review{If $ \Sigma $ is diagonal,
	the $ i $th diagonal element of $ A $ is
	\begin{equation}\label{eq:cons-error-regular-derivative-at-1-trace-elem-diagonal-sigma}
		A_{ii} = \varnode{i} \left(1 - \dfrac{1}{R} + \dfrac{1}{R}\sum_{j\in\regSet}W_{ji}\right) > 0.
	\end{equation}}
	 \item[\review{Condition C2.}] If $ W^o $ symmetric,
	the $ i $th diagonal element of $ A $ is
	\begin{multline}\label{eq:cons-error-regular-derivative-at-1-trace-elem}
		\hspace{-3mm}A_{ii} = \varnode{i} + \dfrac{1}{R}\sum_{m\in\malSet}\!\varnode[i]{m}\!\left(\!1-\hspace{-3mm}\sum_{m'\in\malSet\setminus\{m\}}\!\!\! W_{m'm}^o\!\right)
		- \dfrac{1}{R}\varnode{i}\hspace{-1mm}\sum_{m\in\malSet}\!W_{mi}^o \\
		- \sum_{\substack{j\in\regSet\\j\neq i}}\varnode[i]{j}W_{ij}
		- \dfrac{1}{R}\sum_{\substack{j\in\regSet\\j\neq i}}\varnode[i]{j}\sum_{m\in\malSet}W_{mj}^o - \sum_{m\in\malSet}\varnode[i]{m}W_{im}.
	\end{multline}
	It holds
	\begin{multline}\label{eq:cons-error-regular-derivative-at-1-upper-bound}
		\dfrac{1}{R}\varnode{i}\sum_{m\in\malSet}W_{mi}^o + \sum_{\substack{j\in\regSet\\j\neq i}}\varnode[i]{j}W_{ij}
		+\dfrac{1}{R}\sum_{\substack{j\in\regSet\\j\neq i}}\varnode[i]{j}\sum_{m\in\malSet}W_{mj}^o \\
		+ \sum_{m\in\malSet}\varnode[i]{m}W_{im} \le \max\left\lbrace\dfrac{1}{R}\varnode{i}+\varnode[i]{m^*},\varnode[i]{j^*}\right\rbrace,
	\end{multline}
	where $ j^*\doteq\argmax_{\substack{j\in\regSet\setminus\{i\}}}\varnode[i]{j} $ and $ m^*\doteq\argmax_{m\in\malSet}\varnode[i]{m} $.
	Inequality~\eqref{eq:cons-error-regular-derivative-at-1-upper-bound} can be split into the following two cases.
	\begin{description}
		\item[Case \boldmath$ \frac{1}{R}\varnode{i}+\varnode{im^*} \ge \varnode{ij^*} $:]
		\begin{multline}\label{eq:cons-error-regular-derivative-at-1-upper-bound-case-1}
			A_{ii} \ge \varnode{i} + \dfrac{1}{R}\varnode[i]{m^*} - \dfrac{1}{R}\varnode{i} - \varnode[i]{m^*} \\
			= \left(\varnode{i}-\varnode[i]{m^*}\right)\left(1-\dfrac{1}{R}\right) > 0.
		\end{multline}
		\item[Case \boldmath$ \frac{1}{R}\varnode{i}+\varnode{im^*} < \varnode{ij^*} $:]
		\begin{equation}\label{eq:cons-error-regular-derivative-at-1-upper-bound-case-2}
			\hspace{-3mm}A_{ii} \ge \varnode{i} - \varnode[i]{j^*} + \dfrac{1}{R}\sum_{m\in\malSet}\varnode[i]{m}\left(1-\!\!\sum_{m'\in\malSet\setminus\{m\}}W_{m'm}^o\right) > 0.
		\end{equation}
	\end{description}
	\review{The final inequalities in~\eqref{eq:cons-error-regular-derivative-at-1-upper-bound-case-1}--\eqref{eq:cons-error-regular-derivative-at-1-upper-bound-case-2}
	follow from $\Varn\succ0$ and the Gershgorin circle theorem
	that imply $\varnode{i}>\varnode[i]{j} \ \forall i,j\in\sensSet$.}
\end{description}
It follows that the derivative~\eqref{eq:cons-error-regular-derivative-at-1} is positive
and the consensus error~\eqref{eq:cons-error-regular} is increasing 
in a left neighborhood of $ 1 $.
By continuity of~\eqref{eq:cons-error-regular-derivative},
the minimum points satisfy $ \lam^* < 1 $. 

\subsubsection{Part two: $ {\lam^* > 0} $}
\review{From~\cref{lem:P-monotonic},
the error term $\eregn(\lam)$ has negative right derivative at $\lam=0$.}
By continuity of the derivative of \review{$\eregv(\lam)$},
we can compute the following limit: 
\begin{equation}\label{eq:cons-error-regular-derivative-limit}
	\begin{aligned}
		\lim_{\lam\rightarrow0^+}\review{\dfrac{\d{\eregv}(\lam)}{\d{\lam}}} 
		&= \tr{\Var\lim_{\lam\rightarrow0^+}\dfrac{\d{L}}{\d{\lam}}^\top\selreg^\top \lim_{\lam\rightarrow0^+}E}\\
		&= \tr{\Var\Gamma^\top\selreg^\top\left(\selreg\overline{W}-\Creg\selreg\right)}\\
		&= \tr{\Var\Gamma^\top\selreg^\top\big[\!-\Creg \,| \,C_{RM}\big]}\\
		&= \tr{-\Varn_{11}\Gamma_{1}^\top\Creg-\Varn_{12}\Gamma_{2}^\top\Creg+\right.\\
		&\hspace{3mm} \left.\Varn_{12}^\top\Gamma_{1}^\top C_{RM}+(\Varn_{22}+V)\Gamma_{2}^\top C_{RM}},
	\end{aligned}
\end{equation}
where the steady-state consensus matrix $ \overline{W} \doteq \lim_{\lam\rightarrow0^+}L $ 
%
has block partition (cf.~\cref{ass:mal-node-dynamics} for the value of $ \overline{W} $)
\begin{equation}\label{eq:Wbar-partition}
	\overline{W} = \left[\begin{array}{ c | c }
		0 & C_{RM} \\
		\hline
		0 & I_M
	\end{array}\right].
\end{equation}
Matrix $ \Gamma $ can be computed from the spectral decomposition of $ W $.
In particular, 
its elements are finite, 
$ \Gamma_{1} $ is nonnegative, 
and $ \Gamma_{2} $ is nonpositive (details in~\cref{app:L-derivative-limit}).
Putting together~\eqref{eq:cons-error-regular-derivative-limit} and~\cref{lem:P-monotonic},
the right derivative of $\ereg(\lam)$ at $\lam=0$ is negative if and only if
the following inequality holds,
\begin{multline}\label{eq:opt-lam-greater-zero-condition-proof}
	\review{-\dfrac{\d{\eregn}(0)}{\d{\lam}}} - \tr{V\Gamma_{2}^\top C_{RM}} > \tr{-\Varn_{11}\Gamma_{1}^\top\Creg-\right.\\
	\left.\Varn_{12}\Gamma_{2}^\top\Creg+\Varn_{12}^\top\Gamma_{1}^\top C_{RM}+\Varn_{22}\Gamma_{2}^\top C_{RM}},
\end{multline}
which coincides with~\eqref{eq:opt-lam-greater-zero-condition}.
If~\eqref{eq:opt-lam-greater-zero-condition-proof} holds,
$ \ereg(\lam) $ is strictly decreasing in a right neighborhood of $ \lam = 0 $
and $ \lam^* > 0 $.
\review{If $\Varn$ is diagonal,
then $\Varn_{12}=0$ and~\eqref{eq:opt-lam-greater-zero-condition-proof} is always satisfied.}

\subsection{\titlecap{computation of matrix $ \Gamma $}}\label{app:L-derivative-limit}

We now show how to derive $ \Gamma $ from $ W $
and discuss the sign of its elements.
For the sake of simplicity,
we assume that the nominal weight matrix $ W^o $
is symmetric,
which implies that both $W^o$ and $ W $ are diagonalizable. 
If $ W $ is not diagonalizable,
a similar derivation
(with more tedious but conceptually identical calculations)
can be carried out by considering the Jordan canonical form.
This is because a straightforward extension of~\cref{lem:eigendecomposition-inverse} shows that
$ W $ and $\Gamma$ share the same (chain of) generalized eigenvectors.


\myParagraph{Computation of \boldmath $ \Gamma $}
The derivative of $ L $ is (\cref{lem:derivative-inverse})
\begin{equation}\label{eq:L-inv-derivative}
	\dfrac{dL}{d\lam} = \tilde{L} -\lam \tilde{L}\dfrac{d\tilde{L}\inv}{d\lam}\tilde{L} = \tilde{L}-\lam \tilde{L}W\tilde{L},
\end{equation}
where $ \tilde{L}\doteq\left(I-(1-\lam)W\right)\inv $.
Let $ \lambda_W $ and $ v_W $ an eigenvalue of $ W $ and its associated eigenvector, respectively,
from Lemmas~\ref{lem:eigendecomposition-inverse}--\ref{lem:eigendecomposition-(I-A)} it follows that
$ \tilde{L} $ has eigenvalue $ \left(1-(1-\lam)\lambda_W\right)\inv $ with associated eigenvector $ v_W $.
Hence, straightforward computations yield
\begin{equation}\label{eq:eigendecomposition-derivative-L}
	\dfrac{dL}{d\lam} v_W 
						= \dfrac{1 - \left(1-(1-\lam)\lambda_W\right)\inv\lam\lambda_W}{\left(1-(1-\lam)\lambda_W\right)}v_W.
\end{equation}
In particular, the dominant eigenvector $ v_W = \one $ (associated with $ \lambda_W=1 $)
is in the kernel of $ \nicefrac{dL}{d\lam} $ for any $ \lam $.
As for the other eigenvectors, by letting $ \lam $ go to zero in~\eqref{eq:eigendecomposition-derivative-L}, one gets
\begin{equation}\label{eq:eigendecomposition-derivative-L-lam-0}
	\Gamma v_W = \left(1-\lambda_W\right)\inv v_W.
\end{equation}
Finally, the eigendecomposition of $ \Gamma $
is obtained from eigenvectors $ v_W $
and eigenvalues $ \left(1-\lambda_W\right)\inv $,
plus the kernel. 

\myParagraph{Sign of $ \Gamma_{1} $ and $ \Gamma_{2} $}
As regards $ \Gamma_{1} $,
note that the upper-left block in $ \overline{W} $ is identically zero,
and that $ L $ is a stochastic matrix for any value of $ \lam $:
hence, as $ \lam $ becomes larger than zero,
(some) elements in $ L_1 $ become positive,
and thus their derivative at $ \lam = 0^+ $ is also positive.

As for $ \Gamma_{2} $, 
define the following block partitions,
\begin{equation}\label{eq:block-matrix-L}
	\quad L = \left[\begin{array}{ c | c }
		L_{1} & L_{2} \\ 
		\hline
		0 & I_M
	\end{array}\right],
\end{equation}
with $ W_1,L_{1}\in\Real{R\times R} $ and $ W_2,L_{2}\in\Real{R\times M} $.
Then, it holds
\begin{equation}\label{eq:L_im-decreasing-1}
	\dfrac{dL}{d\lam} = \dfrac{1}{\lam}L\left(I-WL\right) = 
	\left[\begin{array}{ c | c }
		\star & -L_1W_{1} L_{2} - L_{1} \\ 
		\hline
		0 & 0
	\end{array}\right],
\end{equation}
which implies, for any $ \lam\in(0,1) $,
\begin{equation}\label{eq:L_im-decreasing-2}
	\dfrac{dL_{im}}{d\lam} \le 0, \, i\in\regSet, m\in\malSet.
\end{equation}
In particular, the limit of the derivative of element $ L_{im} $ at $ \lam = 0^+ $ is nonpositive 
in virtue of the theorem of sign permanence.

\subsection{Proof of~\cref{prop:error-increases-with-d}}\label{app:proof-error-vs-d}

\begin{description}[leftmargin=*]
	\item[\review{Dependence on \boldmath $V$.}] Note that $\eregn$ is independent of \review{$V$}.
	From~\eqref{eq:cons-error-explicit-1},
	\review{we highlight the contribution of $v$ to the error $ \ereg $ as follows}: 
	\begin{equation}\label{eq:cons-error-regular-partial-d}
		\review{\eregv = \tr{L_{2}VL_{2}^\top} + \kappa,}
	\end{equation}
	\review{where $\kappa$ does not depend on $V$ and we use the block partition
	\begin{equation}
		L = \left[\begin{array}{ c | c }
			L_1 & L_2 \\ 
			\hline
			0 & I
		\end{array}\right].
	\end{equation}
	The matrix $L_2$
	is positive,
	see~\cite{7577815} and discussion in~\autoref{sec:trade-off}).
	Then,
	if $V_1\succ V_2$,
	it follows that $L_{2}V_1L_{2}^\top \succ L_{2}V_2L_{2}^\top$
	and hence the trace in~\eqref{eq:cons-error-regular-partial-d} is strictly increasing with $V$.
	}
	\review{\item[Dependence on \boldmath $Q$.] Note that $\eregv$ is independent of $Q$.
	Let $P_1$ and $P_2$ denote the solutions of~\eqref{eq:P-lyapunov}
	with $Q=Q_1$ and $Q=Q_2$,
	respectively.
	If $Q_1 \succ Q_2$,
	then $\tilde{Q}_1 \succ \tilde{Q}_2$ and
	it is known that $P_1\succ P_2$,
	from which the claim follows.}
\end{description}

\subsection{Proof of~\cref{prop:opt-lam-vs-noise}}\label{app:proof-opt-lam-vs-noise}

\begin{description}[leftmargin=*]
	\item[\review{Dependence on \boldmath $V$.}]
	\review{In the following,
	we make the dependence of the error $\ereg$ on $d_m$ explicit.} 
	Let us compute the partial derivative of the error
	first w.r.t. $ \lam $ and then w.r.t. $ d_m $:
	\begin{equation}\label{eq:cons-error-derivative-lam-d}
		\hspace{-1mm}\dfrac{\partial^2\ereg(\lam,d_m)}{\partial d_m\partial\lam} = 
		\dfrac{1}{\lam}\tr{L\dfrac{\d{\Var(d_m)}}{\d{d_m}} L^\top\left(I-W^\top L^\top\right)\selreg^\top\selreg}.
	\end{equation}
	It holds
	\begin{gather}
		L\dfrac{\d{\Var(d_m)}}{\d{d_m}} = \left[\begin{array}{ c | c }
			0 & L_{2}S_m \\ 
			\hline
			0 & S_m
		\end{array}\right] \label{eq:other-block-matrices}\\
		M \doteq I-W^\top L^\top = \left[\begin{array}{ c | c }
			I_R - W_{1}^\top L_{1}^\top & 0 \\ 
			\hline
			-W_{2}^\top L_{1}^\top - L_{2}^\top & 0
		\end{array}\right] \label{eq:M-block-matrix} \\
		L^\top M\selreg^\top\selreg = \left[\begin{array}{ c | c }
			\star & 0 \\ 
			\hline
			-L_{2}^\top W_{1}^\top L_{1}^\top - W_{2}^\top L_{1}^\top & 0
		\end{array}\right]
	\end{gather}
	and the argument of the trace in~\eqref{eq:cons-error-derivative-lam-d} is
	\begin{equation}\label{eq:cons-error-derivative-lam-d-argument}
		\left[\begin{array}{ c | c }
			-L_{2}S_mL_{2}^\top W_{1}^\top L_{1}^\top - L_{2}S_mW_{2}^\top L_{1}^\top & 0 \\ 
			\hline
			\star & 0
		\end{array}\right]
	\end{equation}
	whose upper-left block is a negative matrix for all $ \lam\in(0,1) $,
	and is the zero matrix for $ \lam = 1 $.
	Hence, the derivative of the consensus error w.r.t. $ \lam $~\eqref{eq:cons-error-regular-derivative}
	is strictly decreasing with $ d_m $ for any $ \lam\in(0,1) $.
	By continuity of~\eqref{eq:cons-error-regular-derivative},
	the minimum points of $ \ereg(\lam) $ are strictly increasing with $ d_m $.
	\review{\item[Dependence on \boldmath $Q$.]
		Note that $\eregv$ is independent of $Q$.
		We consider the derivative of $\eregn$ w.r.t. $ \lam $:
		\begin{equation}\label{eq:der-eregn}
			\dfrac{\d{\eregn}(\lam)}{\d{\lam}} = -\sum_{i\in\regSet}\dfrac{2(1-\lam)\left([\Wreg P\Wreg^\top]_{ii}+\tilde{Q}_{ii}\right)}{1 - (1-\lam)^2W_{ii}^2}.
		\end{equation}
	Let $Q_1\succ Q_2$,
	then it holds $P_1 \succ P_2$,
	which implies $[\Wreg P_1\Wreg^\top]_{ii}> [\Wreg P_2\Wreg^\top]_{ii}$ $\forall i\in\regSet$.
	Further,
	it holds $\tilde{Q}_1 \succ \tilde{Q}_2$ and $[\tilde{Q}_1]_{ii} > [\tilde{Q}_2]_{ii}$ $\forall i\in\regSet$.
	By combining such two facts,
	we conclude that~\eqref{eq:der-eregn} is strictly decreasing.
	The statement follows by the same argument of the case above.}
\end{description}

\subsection{Proof of~\cref{prop:opt-lam-noise-inf}}\label{app:proof-opt-lam-noise-inf}

\begin{description}[leftmargin=*]
	\item[\review{Dependence on \boldmath $V$.}]
	We expand~\eqref{eq:cons-error-regular-derivative} to highlight $ d_m $:
	\begin{equation}\label{eq:cons-error-regular-derivative-dm}
		\begin{aligned}
			\dfrac{\d{\ereg}(\lam)}{\d{\lam}} = \dfrac{N_{mm}}{\lam}d_{m} + \review{\kappa},
		\end{aligned}
	\end{equation}
	where $ N $ is the nonpositive matrix given by
	\begin{equation}\label{eq:N}
		N = -\left(L_{2}^\top W_{1}^\top + W_{2}^\top \right) L_{1}^\top L_{2}
	\end{equation}
	\review{and $ \kappa $ does not depend on $ d_m $}.
	Note that $ N_{mm} \neq 0 $ because the opposite implies that the $ m $th malicious \node has no interactions with regular \nodes.
	It follows that for any $ \lam < 1 $ there exists $ d_m \ge 0 $ such that~\eqref{eq:cons-error-regular-derivative-dm} is negative,
	which is given by
	the following inequality:
	\begin{equation}\label{eq:cons-error-regular-derivative-dm-negative}
		d_m > -\dfrac{\kappa\lam}{N_{mm}} - \sum_{\substack{m'\in\malSet\\m'\neq m}}\dfrac{N_{m'm'}}{N_{mm}}d_{m'}.
	\end{equation}
	The claim follows by combining~\eqref{eq:cons-error-regular-derivative-dm-negative} with~\cref{prop:opt-lam-vs-noise}.
	\review{\item[Dependence on \boldmath $Q$.]
		We use~\eqref{eq:P-partial-lam} to highlight $ q_{m} $ in~\eqref{eq:cons-error-regular-derivative}:
		\begin{equation}\label{eq:cons-error-regular-derivative-Q}
			\dfrac{\d{\ereg}(\lam)}{\d{\lam}} = -\sum_{i\in\regSet}\dfrac{2(1-\lam)\left([\Wreg P\Wreg^\top]_{ii}^2+q_{m}\tilde{Q}_{im}^2\right)}{1 - (1-\lam)^2W_{ii}^2}  + \kappa,
		\end{equation}
		where $\kappa$ does not depend on $q_{m}$.
		Also,
		$P$ is increasing with $q_{m}$ and thus $[\Wreg P_1\Wreg^\top]_{ii}$ also is.
		Hence,
		for any $\lam<1$,
		there exists $q_{m}\ge0$ such that~\eqref{eq:cons-error-regular-derivative-Q} is negative,
		which is given by the following inequality:
		\begin{equation}\label{eq:cons-error-regular-derivative-Q-negative}
			\sum_{i\in\regSet}\dfrac{2(1-\lam)\left([\Wreg P\Wreg^\top]_{ii}^2+q_{m}\tilde{Q}_{im}^2\right)}{1 - (1-\lam)^2W_{ii}^2} > \kappa.
		\end{equation}
		The claim follows by combining~\eqref{eq:cons-error-regular-derivative-Q-negative} with~\cref{prop:opt-lam-vs-noise}.
	}
\end{description}
	

\begin{IEEEbiography}[{\includegraphics[width=1in,height=1.25in,clip,keepaspectratio]{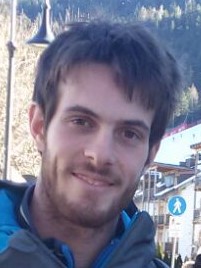}}]{Luca Ballotta}
	received the Master's degree in Automation Engineering and the Ph.D. degree in Information Engineering from the University of Padova, Italy, in 2019 and 2023, respectively.
He is currently a research fellow at the University of Padova, Department of Information Engineering.
He was Visiting Student at the Massachusetts Institute of Technology in 2020 and 2022.
He was awarded with the Young Author Prize at the 2020 IFAC World Congress.
His research interests include multi-agent systems and networked control systems under resource constraints, 
resilient distributed optimization,
and learning-based safe control.

\end{IEEEbiography}

\begin{IEEEbiography}[{\includegraphics[width=1in,height=1.25in,clip,keepaspectratio]{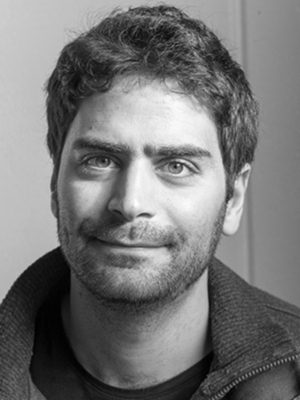}}]{Giacomo Como}
	(Member, IEEE) received the B.Sc., M.S., and Ph.D. degrees in 
applied mathematics from Politecnico di Torino, Turin, Italy, in 2002, 2004, and 2008,
respectively. He is currently a Professor with the Department of Mathematical Sciences, Politecnico 
di Torino. He is also a Senior Lecturer with the Automatic Control Department, Lund University, Lund, Sweden. He was a Visiting Assistant in research with Yale University, New
Haven, CT, USA, in 2006–2007 and a Postdoctoral Associate with the Laboratory for
Information and Decision Systems, Massachusetts Institute of Technology, Cambridge, 
MA, USA, from 2008 to 2011. His research interests include dynamics, information, and 
control in network systems with applications to cyberphysical systems, infrastructure 
networks, and social and economic networks. Dr. Como is currently a Senior Editor for  IEEE TRANSACTIONS ON 
CONTROL OF NETWORK SYSTEMS, an Associate Editor for Automatica and the Chair of the IEEE-CSS Technical Committee on Networks and Communications. He was an Associate
Editor for IEEE TRANSACTIONS ON NETWORK SCIENCE AND ENGINEERING (2015–2021) and IEEE TRANSACTIONS ON 
CONTROL OF NETWORK SYSTEMS (2016–2022). He was the IPC Chair of the IFAC Workshop NecSys’15 and a
Semiplenary Speaker at the International Symposium MTNS’16. He was the recipient of the 2015 George S. Axelby Outstanding Paper Award. 
\end{IEEEbiography}

\begin{IEEEbiography}[{\includegraphics[width=1in,height=1.25in,clip,keepaspectratio]{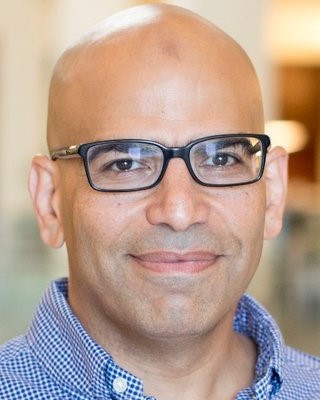}}]{Jeff S. Shamma}
	is the Department Head \& Dobrovolny Chair, Industrial and Enterprise
Systems Engineering, University of Illinois Urbana-Champaign. He is the former Director of the Center
of Excellence for NEOM Research and a Professor
of Electrical Engineering at King Abdullah University of Science and Technology (KAUST). Before
joining KAUST, Shamma was Julian T. Hightower
Chair in Systems \& Control in the School of Electrical and Computer Engineering at Georgia Tech,
and has held faculty positions at the University of
Minnesota, The University of Texas at Austin, and the University of California,
Los Angeles. Shamma received a Ph.D. in systems science and engineering
from MIT in 1988. He is the recipient of an NSF Young Investigator Award,
the American Automatic Control Council Donald P. Eckman Award, and the
Mohammed Dahleh Award, and he is a Fellow of the IEEE and of the IFAC
(International Federation of Automatic Control). 
Shamma is the Editor-in-Chief for the IEEE Transactions on Control of Network Systems

\end{IEEEbiography}

\begin{IEEEbiography}[{\includegraphics[width=1in,height=1.25in,clip,keepaspectratio]{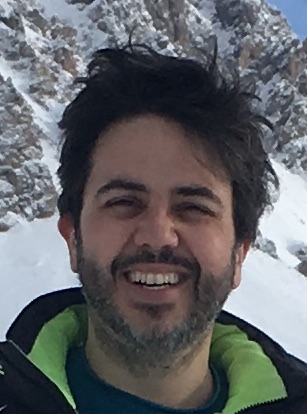}}]{Luca Schenato}
	received the Dr. Eng. degree in electrical engineering from the University of Padova in 1999 and the Ph.D. degree in Electrical Engineering and Computer Sciences from the UC Berkeley, in 2003. He held a post-doctoral position in 2004 and a visiting professor position in 2013-2014 at U.C. Berkeley. Currently he is Full Professor at the Information Engineering Department at the University of Padova. His interests include networked control systems, multi-agent systems, wireless sensor networks, smart grids and cooperative robotics. Luca Schenato has been awarded the 2004 Researchers Mobility Fellowship by the Italian Ministry of Education, University and Research (MIUR), the 2006 Eli Jury Award in U.C. Berkeley and the EUCA European Control Award in 2014, and IEEE Fellow in 2017. He served as Associate Editor for IEEE Trans. on Automatic Control from 2010 to 2014 and he is he is currently Senior Editor for IEEE Trans. on Control of Network Systems and Associate Editor for Automatica.

\end{IEEEbiography}
\end{document}